\documentclass[12pt]{article}
\usepackage{amssymb}
\usepackage{amsmath}
\usepackage{amsthm}
\usepackage{color}
\usepackage{comment}
\usepackage{geometry}		% allows easy margin settings
\usepackage{graphicx}

\makeatletter
	\@addtoreset{equation}{section}
\makeatother
\renewcommand{\labelenumi}{\roman{enumi})}

\DeclareFontFamily{OT1}{pzc}{}
\DeclareFontShape{OT1}{pzc}{m}{it}{<-> s * [1.10] pzcmi7t}{}
\DeclareMathAlphabet{\mathpzc}{OT1}{pzc}{m}{it}

\let\originalleft\left
\let\originalright\right
\renewcommand{\left}{\mathopen{}\mathclose\bgroup\originalleft}
\renewcommand{\right}{\aftergroup\egroup\originalright}

\begin{document}

\newcommand{\bO}{{\bf 0}}
\newcommand{\co}{\mathpzc{o}}
\newcommand{\rD}{{\rm D}}
\newcommand{\ee}{\varepsilon}
\newcommand{\ri}{{\rm i}}

\newtheorem{theorem}{Theorem}[section]
\newtheorem{corollary}[theorem]{Corollary}
\newtheorem{lemma}[theorem]{Lemma}
\newtheorem{proposition}[theorem]{Proposition}

\theoremstyle{definition}
\newtheorem{definition}{Definition}[section]
\newtheorem{assumption}[definition]{Assumption}
%\newtheorem{example}[definition]{Example}

%\theoremstyle{remark}
%\newtheorem{remark}{Remark}[section]

% To do:
%		- IG to create figures, draft overview, and discussion section
%		- send to RMcL
%		- edit more, check formulas carefully (particularly that inequalities are correctly stated as strict vs non-strict)

% Style comments:
%		- spelling: GB English
%		-	max figure width used for Nonlinearity (IOP): one column, 15.5cm
%		- abbreviations: BCNF

%\title{Explicit constructions for robust chaos in $\mathbb{R}^n$}
\title{Robust chaos in $\mathbb{R}^n$}
\author{I.~Ghosh, D.J.W.~Simpson\\\\
School of Mathematical and Computational Sciences\\
Massey University\\
Palmerston North, 4410\\
New Zealand}

\maketitle

% keywords: piecewise-linear; piecewise-smooth; border-collision bifurcation; chaotic attractor; invariant expanding cone
% MSC codes: 37G35; 39A28

\begin{abstract}

We treat $n$-dimensional piecewise-linear continuous maps with two pieces,
each of which has exactly one unstable direction,
and identify an explicit set of sufficient conditions
for the existence of a chaotic attractor.
The conditions correspond to an open set within the space of all such maps,
allow all $n \ge 2$,
and allow all possible values for the unstable eigenvalues in the limit that all stable eigenvalues tend to zero.
To prove an attractor exists we use the stable manifold of a fixed point to construct a trapping region;
to prove the attractor is chaotic we use the unstable directions to construct
an invariant expanding cone for the derivatives of the pieces of the map.
We also show the chaotic attractor is persistent under nonlinear perturbations,
thus when such an attractor is created locally in a border-collision bifurcation of a general piecewise-smooth system,
it persists and is chaotic for an interval of parameter values beyond the bifurcation.

\end{abstract}

%===============================================================================
\section{Introduction}
\label{sec:intro}

Robust chaos refers to the existence of a chaotic attractor for an open subfamily
of a family of dynamical systems \cite{Gl17}. % [ZeSp12]
The Lorenz system exhibits robust chaos \cite{Tu99,Tu02},
as do smooth ODEs with Lorenz-type attractors \cite{GuWi79}.
In contrast, robust chaos does not occur for the logistic family and smooth unimodal maps more broadly
because for these maps periodic windows are dense in parameter space \cite{GrSw97,Ly97,Va10}.

Robust chaos is prevalent in families of piecewise-smooth maps
because the absence of a smooth turning point inhibits periodic windows.
If a piecewise-smooth map is expanding (in every direction)
then under mild technical conditions any attractors must be chaotic \cite{Bu99,Co02,Ts01}.
Such expansion is a relatively strong property;
in this paper we consider maps with a single direction of instability.

For maps with a single direction of instability, ergodic theory techniques
provide abstract sets of constraints that ensure robust chaos \cite{Ry04,WaYo01,WaYo08,Yo85}.
An alternate approach is to identify {\em explicit} conditions on particular families of maps.
Misiurewicz \cite{Mi80} achieved this for the Lozi family
\begin{equation}
x \mapsto \begin{bmatrix} -\tilde{a} |x_1| + x_2 + 1 \\ \tilde{b} x_1 \end{bmatrix},
\label{eq:Lozi}
\end{equation}
where $x = (x_1,x_2) \in \mathbb{R}^2$ and tildes have been added to reduce confusion to our later notation.
His results show that if $\tilde{b} > 0$ and $\tilde{b}+1 < \tilde{a} < 2 - \frac{\tilde{b}}{2}$
then \eqref{eq:Lozi} has a chaotic attractor.
This was achieved by showing that over the given parameter range the map has a trapping region (a compact set that maps to its interior) and hence a topological attractor.
Misiurewicz also constructed a cone of tangent vectors that is invariant and expanding
under multiplication by the derivative of each piece of the map.
This implies the attractor is chaotic in the sense of a positive Lyapunov exponent.
Further constraints on the parameters ensure the attractor displays other aspects of chaos
(e.g.~mixing and uniform hyperbolicity) \cite{Mi80}.

The Lozi family was designed as a prototypical family of invertible maps capable of exhibiting chaos \cite{Lo78}.
To quantitatively mimic the behaviour of physical systems,
it is necessary to generalise the Lozi family to 
\begin{equation}
f(x) = \begin{cases}
A_L x + b, & x_1 \le 0, \\
A_R x + b, & x_1 \ge 0,
\end{cases}
\label{eq:bcnf}
\end{equation}
where
\begin{align}
A_L &= \begin{bmatrix}
\tau_L & 1 \\
-\delta_L & 0
\end{bmatrix}, &
A_R &= \begin{bmatrix}
\tau_R & 1 \\
-\delta_R & 0
\end{bmatrix}, &
b = \begin{bmatrix} 1 \\ 0 \end{bmatrix}.
\label{eq:bcnf2dALARb}
\end{align}
This is the two-dimensional border-collision normal form (BCNF) of Nusse and Yorke \cite{NuYo92},
except we have omitted the border-collision parameter $\mu$
(reintroduced in \S\ref{sec:bcb} where we describe border-collision bifurcations).
There are four parameters, $\tau_L, \delta_L, \tau_R, \delta_R \in \mathbb{R}$;
any continuous, two-piece, piecewise-linear map on $\mathbb{R}^2$
satisfying a certain non-degeneracy condition can be converted to this form via an affine change of coordinates.

Banerjee {\em et al.}~\cite{BaYo98} argued that if
$0 < \delta_L < \tau_L - 1$, $0 < \delta_R < -\tau_R - 1$, and
\begin{equation}
\delta_R - (1 + \tau_R) \delta_L + \frac{1}{2} \big( (1 + \tau_R) \tau_L - \tau_R - \delta_L - \delta_R \big)
\left( \tau_L + \sqrt{\tau_L^2 - 4 \delta_L} \right) > 0,
\nonumber
\end{equation}
then the two-dimensional BCNF has a chaotic attractor and
this was later verified rigorously using Misiurewicz's techniques \cite{GlSi21}.
Our previous work \cite{GhMc23} showed that robust chaos
extends to a larger parameter region including
areas where the map is orientation-reversing and non-invertible.
We stress that this region, and that of Misiurewicz \cite{Mi80},
are open regions of parameter space where the map has a single direction of instability.
Other regions of robust chaos can be identified rigorously
by using more complex trapping region constructions
and applying cones to higher iterates of the map \cite{GlSi22b,Si23e}.

However, this avenue of research has thus far only dealt with two-dimensional maps.
The purpose of this paper is to extend the constructions to higher dimensions.
This task is essential for the results to be applied to mathematical models
for which dimensionality is an artifact of modelling assumptions.
Many physical systems have dynamics well-modelled by
piecewise-linear maps, and perturbations of such maps,
and for many of these chaotic dynamics serves an important role.
Vibro-impact energy harvesters use chaos to stabilise high-energy
operational modes with a small control force \cite{KuAl16},
optical resonators use chaotic photon trajectories
to enable additional energy storage \cite{LiDi13},
and cryptography algorithms use chaotic maps
for message encryption \cite{Fr98,KoLi11,PaPa06}.
For such encryption the chaotic parameter region
corresponds to the key space of the algorithm
which needs to be large and devoid of periodic windows
as these provide an avenue for attack \cite{AlMo03,Ko01}.

The extension to higher dimensions is challenging for several reasons.
It is necessary to work with the $n$-dimensional BCNF \cite{Di03,Si16}
which consists of \eqref{eq:bcnf} with $x = (x_1,x_2,\ldots,x_n) \in \mathbb{R}^n$ and
\begin{align}
A_L &= \begin{bmatrix}
-a^L_1 & 1 \\
-a^L_2 && \ddots \\
\vdots &&& 1 \\
-a^L_n
\end{bmatrix}, &
A_R &= \begin{bmatrix}
-a^R_1 & 1 \\
-a^R_2 && \ddots \\
\vdots &&& 1 \\
-a^R_n
\end{bmatrix}, &
b = \begin{bmatrix} 1 \\ 0 \\ \vdots \\ 0 \end{bmatrix}.
\label{eq:bcnfALARb}
\end{align}
Parameter space is $\mathbb{R}^{2 n}$ because
each $a^L_i$ and $a^R_i$, where $i = 1,\ldots,n$,
can take any value in $\mathbb{R}$.
Trapping region constructions for the two-dimensional BCNF have used convex polygons
$\Omega$ because the image of a polygon under \eqref{eq:bcnf} is another polygon,
and if the vertices of the image of $\Omega$ belong to the interior of $\Omega$
then the entire image belongs to the interior,
so $\Omega$ is a trapping region.
Below we use the same idea in higher dimensions:~we construct an $n$-dimensional convex polytope $\Omega$
and show that if certain conditions are satisfied then the vertices of its image
belong to the interior of $\Omega$.

The only previous work we are aware of where this was achieved with no restriction on $n$
is that of Glendinning \cite{Gl15b} who identified a cuboid that maps into itself.
Glendinning treated parameter values where some iterate of the $n$-dimensional BCNF is expanding,
so a cone was not needed; also the attractor itself was roughly a cube so the cuboid
construction was effective for this parameter region.

Below we assume $A_L$ and $A_R$ have unstable eigenvalues $\alpha > 1$ and $-\beta < -1$, respectively,
and that all other eigenvalues of $A_L$ and $A_R$ have modulus less than $1$.
In the limit that all stable eigenvalues are zero,
\eqref{eq:bcnf} reduces to a one-dimensional map
that has an attractor if and only if $\frac{1}{\alpha} + \frac{1}{\beta} > 1$.
We design $\Omega$ so that for {\em all} such values of $\alpha$ and $\beta$,
the set $\Omega$ is a trapping region if the stable eigenvalues have modulus at most $r$,
and we provide an explicit bound on the value of $r$.
This is challenging because the size of the immediate basin of attraction
tends to zero as $\frac{1}{\alpha} + \frac{1}{\beta} \to 1$ and $r \to 0$.

To verify chaos in the two-dimensional setting,
previous works \cite{Mi80,GlSi21} define the cone to be all points on and between the unstable eigenspaces of $A_L$ and $A_R$.
This approach cannot be used in higher dimensions
because these eigenspaces are one-dimensional and the cone boundary needs to be $(n-1)$-dimensional,
so instead we introduce a sequence of values to define the boundary of a cone $\Psi$.
We design $\Psi$ so that for values of $\alpha$ and $\beta$ in the range discussed above,
it is invariant and expanding if $r$ is sufficiently small.
Again this is challenging because the degree of expansion tends to $1$ as $\min[\alpha,\beta] \to 1$ and $r \to 0$.

We also show the chaotic attractor is robust to nonlinear perturbations to the pieces of the map.
This enables the results to be applied to border-collision bifurcations of generic piecewise-smooth maps
that occur when a fixed point collides with a switching manifold of the map \cite{DiBu08,Si16}.
The robustness is demonstrated by showing that $\Psi$ satisfies the stronger property of being contracting-invariant
and performing coordinate transformations and a blow-up of phase space to
relate the local dynamics associated with the border-collision bifurcation to those of the BCNF.

%-------------------------------------------------------------------------------
\subsection{Overview}

We start in \S\ref{sec:main} by computing fixed points of the $n$-dimensional BCNF
and treating the one-dimensional limit for which the dynamics are governed a skew tent map.
We then state the main result (Theorem \ref{th:main}) and illustrate it numerically with two, three, and ten-dimensional examples.

In \S\ref{sec:basics} we provide formal definitions for the main concepts
and detail the constructions of our polytope $\Omega$ and cone $\Psi$.
In \S\ref{sec:bounds} we perform an array of preliminary calculations,
such as relating the parameters of BCNF to the eigenvalues of $A_L$ and $A_R$
and bounding key quantities.
It is important for these bounds to be simple so that our
later main estimates are analytically tractable.
Sections \ref{sec:forwardInvariantRegion} and \ref{sec:cone} contain the main computations;
we prove $\Omega$ is forward invariant (under certain conditions)
and can be perturbed to a trapping region,
then prove $\Psi$ is contracting-invariant and expanding,
and lastly prove Theorem \ref{th:main}.

In \S\ref{sec:bcb} we consider border-collision bifurcations of fixed points of piecewise-smooth continuous maps.
We show that if the local piecewise-linear approximation to the map
can be converted to the BCNF where it satisfies the conditions of Theorem \ref{th:main},
then the border-collision bifurcation creates a robust chaotic attractor.
This attractor typically grows asymptotically linearly in size as parameters are varied,
and we illustrate this with a three-dimensional example.
Finally \S\ref{sec:discussion} provides conclusions and an outlook for future work.

%-------------------------------------------------------------------------------
\subsection{Notation}

Here we list key notational conventions.
We write $\bO$ for the zero vector (origin) in $\mathbb{R}^n$,
$\| \cdot \|$ for the Euclidean norm,
and $\subset$ for subset.
We use subscripts for vector components,
e.g.~for a point $P \in \mathbb{R}^n$ we write $P = (P_1,\ldots,P_n)$.
Given integers $0 \le k \le n$,
\begin{equation}
\begin{pmatrix} n \\ k \end{pmatrix} = \dfrac{n!}{k! \,(n-k)!}
\nonumber
\end{equation}
denotes the `$n$ choose $k$' binomial coefficient.

The left and right pieces of \eqref{eq:bcnf} are
\begin{align}
f_L(x) &= A_L x + b, &
f_R(x) &= A_R x + b,
\label{eq:fLfR}
\end{align}
and we let
\begin{equation}
\Sigma = \left\{ x \in \mathbb{R}^n \,\middle|\, x_1 = 0 \right\}
\label{eq:Sigma}
\end{equation}
denote the switching manifold.
We refer to
$\left\{ x \in \mathbb{R}^n \,\middle|\, x_1 < 0 \right\}$
and $\left\{ x \in \mathbb{R}^n \,\middle|\, x_1 > 0 \right\}$
as the {\em left half-space} and the {\em right half-space}, respectively.

%===============================================================================
\section{Sufficient conditions for robust chaos}
\label{sec:main}
	
Below we use the eigenvalues of $A_L$ and $A_R$ to refer to an instance of
the $n$-dimensional BCNF \eqref{eq:bcnf}, $f$.
We can do this because the parameters of $f$
are the coefficients of the characteristic polynomials of $A_L$ and $A_R$:
\begin{equation}
\begin{split}
\det(\lambda I - A_L) &= \lambda^n + a^L_1 \lambda^{n-1} + \cdots + a^L_{n-1} \lambda + a^L_n \,, \\
\det(\lambda I - A_R) &= \lambda^n + a^R_1 \lambda^{n-1} + \cdots + a^R_{n-1} \lambda + a^R_n \,.
\end{split}
\label{eq:charPolys}
\end{equation}
Specifically, given any two sets of $n$ numbers, with complex numbers appearing in complex conjugate pairs,
there exists a unique real-valued choice for the parameters of $f$
for which these numbers are the eigenvalues of $A_L$ and $A_R$.

%-------------------------------------------------------------------------------
\subsection{One direction of expansion}
\label{sub:one}

Eigenvalues with modulus greater than $1$ correspond to unstable directions of the dynamics,
while eigenvalues with modulus less than $1$ correspond to stable directions. % \cite{El08,Ma99}.
Below we assume $A_L$ and $A_R$ each have exactly one unstable eigenvalue.
To motivate our assumption on the signs of these eigenvalues,
observe that if $1$ is not an eigenvalue of $A_L$,
then $I - A_L$ is invertible and
\begin{equation}
Y = (I - A_L)^{-1} b
\label{eq:Y}
\end{equation}
is the unique fixed point of $f_L$.
Similarly if $1$ is not an eigenvalue of $A_R$,
then $I - A_R$ is invertible and
\begin{equation}
X = (I - A_R)^{-1} b
\label{eq:X}
\end{equation}
is the unique fixed point of $f_R$.
As in Banerjee {\em et al.}~\cite{BaYo98}, which dealt with the two-dimensional setting,
we consider parameter values for which $Y_1 < 0$ and $X_1 > 0$ so that $Y$ and $X$ are both fixed points of $f$.
Indeed, if instead $Y_1 > 0$ and $X_1 < 0$, then $f$ has no bounded invariant set \cite{Si24e}, so certainly no chaotic attractor.
Straight-forward calculations show that
if $A_L$ and $A_R$ have exactly one unstable eigenvalue,
then $Y_1 < 0$ if and only if the unstable eigenvalue of $A_L$ is greater than $1$,
while $X_1 > 0$ if and only if the unstable eigenvalue of $A_R$ is less than $-1$.
Consequently we make the following assumption.

%...................................................................................................
\begin{assumption}
Consider \eqref{eq:bcnf} with \eqref{eq:bcnfALARb} and $n \ge 2$.
Suppose $\alpha > 1$ is an eigenvalue of $A_L$ with multiplicity one,
$-\beta < -1$ is an eigenvalue of $A_R$ with multiplicity one,
and all other eigenvalues of $A_L$ and $A_R$ have modulus at most $0 < r < 1$.
\label{as:main}
\end{assumption}

%-------------------------------------------------------------------------------
\subsection{Skew tent maps}
\label{sub:skewTentMaps}

It is instructive to consider the limit $r \to 0$ in Assumption \ref{as:main}.
In this case $\alpha$ and $-\beta$ are the only non-zero eigenvalues of $A_L$ and $A_R$,
so $a^L_1 = -\alpha$, $a^R_1 = \beta$, and all other entries in the first columns of $A_L$ and $A_R$ are zero.
In this case the forward orbit of any point in $\mathbb{R}^n$
becomes constrained to the $x_1$-axis (within at most $n-1$ iterations),
then evolves according to the skew tent map
\begin{equation}
x_1 \mapsto \begin{cases}
1 + \alpha x_1, & x_1 \le 0, \\
1 - \beta x_1, & x_1 \ge 0.
\end{cases}
\label{eq:skewTentMap}
\end{equation}
As a two-parameter family,
\eqref{eq:skewTentMap} is the BCNF in one dimension \cite{MaMa93,NuYo95}.
Fig.~\ref{fig:skewTent} shows the map with typical values of $\alpha, \beta > 1$.
This figure highlights the fixed point of the left piece of the map
\begin{equation}
w = \frac{-1}{\alpha-1},
\label{eq:fixedPoint}
\end{equation}
and its preimage under the right piece of the map
\begin{equation}
z = \frac{\alpha}{(\alpha - 1) \beta}.
\label{eq:preimage}
\end{equation}
If $z > 1$ then the interval $[w,z]$ is forward invariant and the map has an attractor.
This attractor is chaotic because both slopes are greater than one in absolute value \cite{LiYo78},
and is a union of $2^m$ intervals, for some $m \ge 1$, see \cite{ItTa79}.
Notice $z > 1$ is equivalent to $\frac{1}{\alpha} + \frac{1}{\beta} > 1$.
If instead $z < 1$ the forward orbit of the critical point $0$ diverges and \eqref{eq:skewTentMap} has no attractor.

In summary, the skew tent map \eqref{eq:skewTentMap} has a chaotic attractor if $\alpha > 1$, $\beta > 1$,
and $\frac{1}{\alpha} + \frac{1}{\beta} > 1$.
Beyond $\frac{1}{\alpha} + \frac{1}{\beta} = 1$ the map has no attractor,
beyond $\beta = 1$ the attractor of the map is a fixed point, so is not chaotic,
while beyond $\alpha = 1$ the map still has a chaotic attractor
but this requires more work to establish \cite{ItTa79b}.

\begin{figure}[h]
\centering
\includegraphics[width=0.5\linewidth]{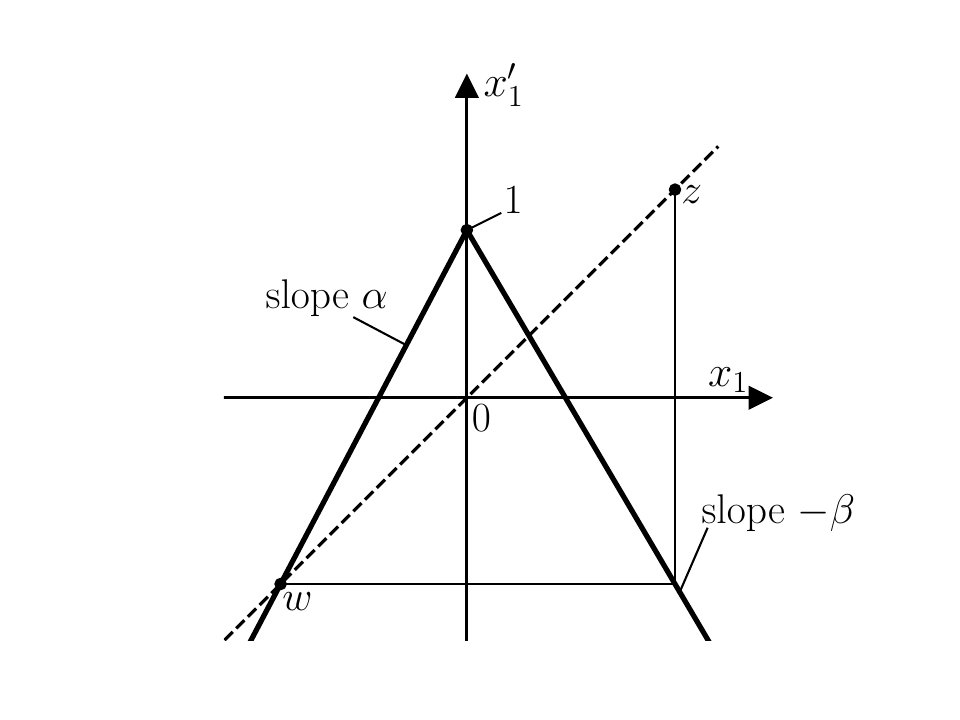}%&
\caption{A sketch of the skew tent map \eqref{eq:skewTentMap}
with slopes $\alpha = 1.9$ and $-\beta = -1.7$.}\label{fig:skewTent}
\end{figure}

%-------------------------------------------------------------------------------
\subsection{Main result}
\label{sub:main}

We now state our main result.
This result is proved in \S\ref{sub:PsiFinal}.

%...................................................................................................
\begin{theorem}
\label{thm:main}
If Assumption \ref{as:main} holds and
\begin{align}
r(n-1) &< \frac{3}{7} \left( 1 - \frac{1}{\alpha} \right), \label{eq:rA} \\
r(n-1) &< \frac{3}{7} \left( 1 - \frac{1}{\beta} \right), \label{eq:rB} \\
r(n-1) &< \frac{1}{10} \left( \frac{1}{\alpha} + \frac{1}{\beta} - 1 \right), \label{eq:rC}
\end{align}
then $f$ has a topological attractor with a positive Lyapunov exponent.
\label{th:main}
\end{theorem}

Lyapunov exponents represent the average rate of expansion of arbitrarily close forward orbits \cite{Vi14}.
Positive Lyapunov exponents correspond to expansion within an attracting object
and are a standard indicator of chaos \cite{Me07}.
Fig.~\ref{fig:region1_3D} shows the upper bound on $r(n-1)$ specified by Theorem \ref{th:main}.
This figure uses $\frac{1}{\alpha}$ and $\frac{1}{\beta}$ on the axes
so that the domain is finite and the bound is piecewise-linear.

\begin{figure}[h]
\centering
\includegraphics[width=0.5\linewidth]{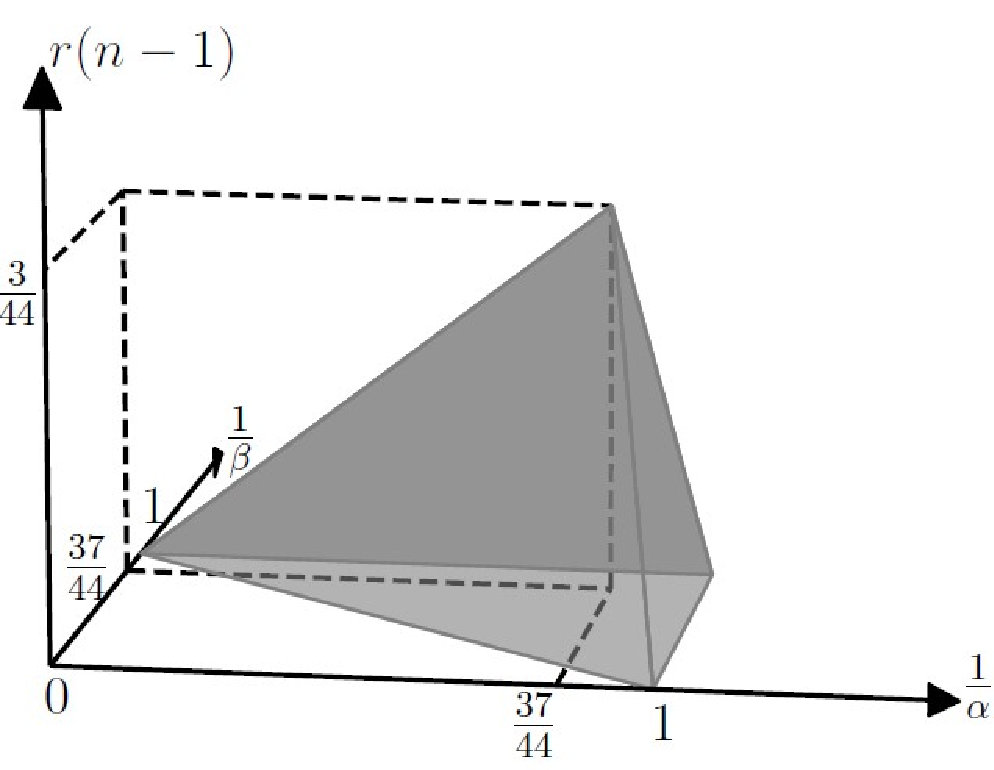}%&
\caption{
The shaded region indicates values of $\frac{1}{\alpha}$, $\frac{1}{\beta}$,
and $r(n-1)$ satisfying \eqref{eq:rA}, \eqref{eq:rB}, and \eqref{eq:rC}.
}\label{fig:bound}
\label{fig:region1_3D}
\end{figure}

The conditions \eqref{eq:rA}, \eqref{eq:rB}, and \eqref{eq:rC}
define an open region of $2 n$-dimensional parameter space and
Theorem \ref{th:main} shows that the BCNF exhibits robust chaos over this region.
With more work one could likely weaken the conditions
and verify robust chaos over a larger parameter region,
but in two respects Theorem \ref{th:main} cannot be improved upon.
The attractor of the skew tent map \eqref{eq:skewTentMap} is destroyed at $\frac{1}{\alpha} + \frac{1}{\beta} = 1$,
thus the bound on $r$ must go to zero as $\frac{1}{\alpha} + \frac{1}{\beta} \to 1$.
Similarly the attractor of the skew tent map \eqref{eq:skewTentMap} becomes non-chaotic at $\beta = 1$,
so the bound on $r$ must go to zero as $\beta \to 1$.
It should be stressed that significant effort was required to produce a result
for which the bound on $r$ tends to zero at $\beta = 1$ and $\frac{1}{\alpha} + \frac{1}{\beta} = 1$,
rather than before these boundaries are reached,
and this is reflected in the complexity of the calculations in \S\ref{sub:Qs} and \S\ref{sub:contractingInvariant}.

%-------------------------------------------------------------------------------
\subsection{Comparison to numerical simulations}
\label{sub:numerics}

Fig.~\ref{fig:pp3D} shows a typical phase portrait using $n = 3$.
Specifically this figure uses
\begin{equation}
\begin{aligned}
a^L_1 &= -1.4, & \qquad \qquad a^R_1 &= 1.45, \\
a^L_2 &= -0.0004, & \qquad \qquad a^R_2 &= 0.0004, \\
a^L_3 &= 0.00056, & \qquad \qquad a^R_3 &= 0.00058,
\end{aligned}
\label{eq:3dexample}
\end{equation}
with which $A_L$ has eigenvalues $\alpha = 1.4$,
$\lambda^L_2 = 0.02$, and $\lambda^L_3 = -0.02$,
while $A_R$ has eigenvalues $-\beta = -1.45$,
$\lambda^R_2 = 0.02 \ri$, and $\lambda^R_3 = -0.02 \ri$.
The assumptions of Theorem \ref{th:main} are satisfied because
\eqref{eq:rA}, \eqref{eq:rB}, and \eqref{eq:rC} hold with $r = 0.02$.
In Fig.~\ref{fig:pp3D} the black dots are $3000$ consecutive iterates of one forward orbit, with transient dynamics removed,
and represent the attractor of the map.
The fixed points $X$ and $Y$ are saddles with one direction of instability.
The attractor is distant from $Y$ but appears to contain $X$,
as is often the case in the two-dimensional setting \cite{GlSi21,GhSi22b}.

\begin{figure}[h]
\centering
\includegraphics[width=0.6\linewidth]{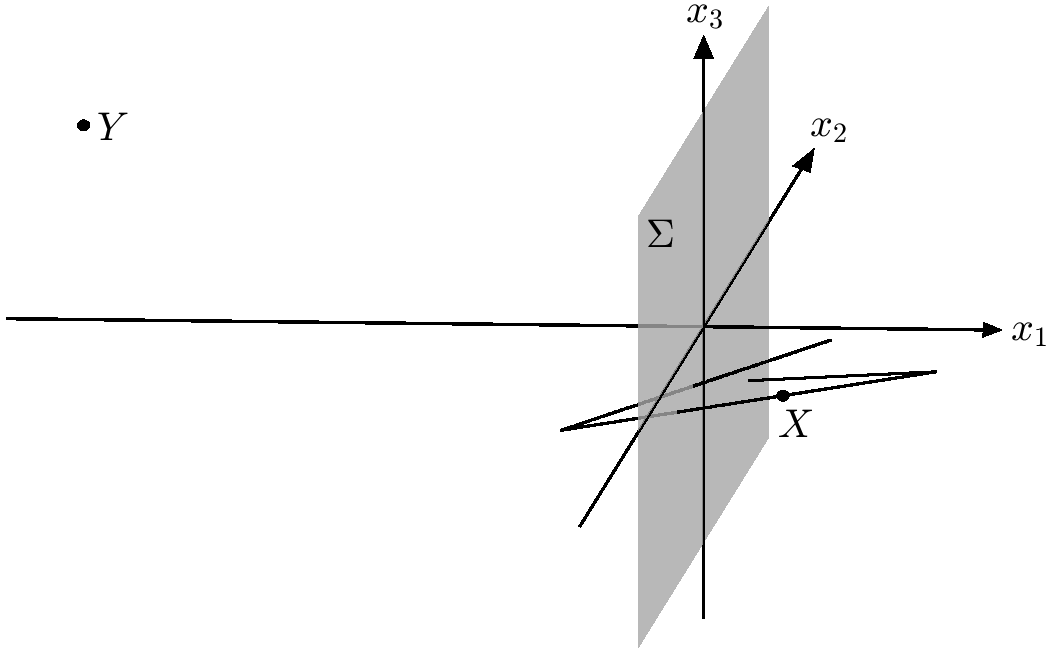}%&
\caption{
A phase portrait of the three-dimensional BCNF, \eqref{eq:bcnf} with
\eqref{eq:bcnfALARb}, with parameter values \eqref{eq:3dexample}.
The attractor is shown in black.
The points $X$ and $Y$ are the saddle fixed points \eqref{eq:Y} and \eqref{eq:X}.
}\label{fig:pp3D}
\end{figure}

Fig.~\ref{fig:pp10D} uses $n = 10$
and parameters for which $A_L$ and $A_R$ have eigenvalues
\begin{equation}
\begin{aligned}
\alpha &= 1.3, & \qquad \qquad -\beta &= -1.7, \\
\lambda^L_{k+2} &= r e^{\frac{2ki \pi}{n-1}},
& \qquad \qquad \lambda^R_{k+2} &= r e^{\frac{(2 k + 1) i \pi}{n-1}},
\end{aligned}
\label{eq:10dexample}
\end{equation}
for $k = 0,1\ldots,8$ and $r = 0.0039$.
As with Fig.~\ref{fig:pp3D} the attractor has a strongly one-dimensional
geometry because all stable eigenvalues are relatively small.
Note that the value of $r$ was chosen slightly less than the largest
value permitted by \eqref{eq:rA}, \eqref{eq:rB}, and \eqref{eq:rC}.
Again, the assumptions of Theorem \ref{th:main} are satisfied because
the conditions \eqref{eq:rA}, \eqref{eq:rB}, and \eqref{eq:rC} all hold with $r=0.0039$.

\begin{figure}[h]
\centering
\includegraphics[width=0.6\linewidth]{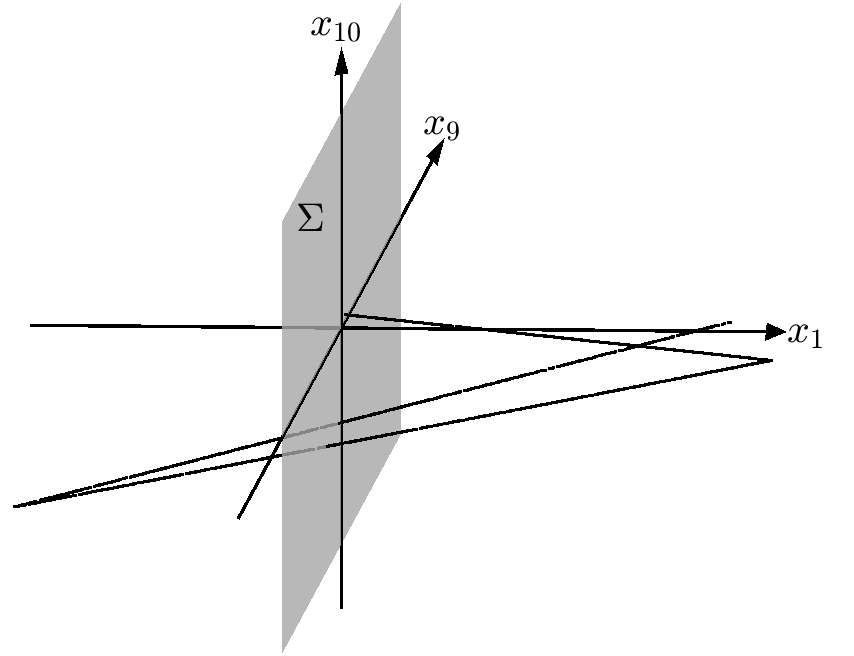}%&
\caption{
A phase portrait of a ten-dimensional
instance of the BCNF \eqref{eq:bcnf} with \eqref{eq:bcnfALARb}.
The parameter values are such that the eigenvalues of $A_L$ and $A_R$
are given by \eqref{eq:10dexample} with $r = 0.0039$.
}\label{fig:pp10D}
\end{figure}

Finally we revisit the two-dimensional setting
and compare Theorem \ref{th:main} to stronger results of earlier papers.
As an example we fix
\begin{align}
\alpha &= 1.3, & \beta &= 1.3,
\label{eq:2dexample}
\end{align}
and consider a range of values for the stable eigenvalues of $A_L$ and $A_R$,
denoted $\lambda^L_2$ and $\lambda^R_2$ respectively.

\begin{figure}[h]
\centering
\includegraphics[width=.5\linewidth]{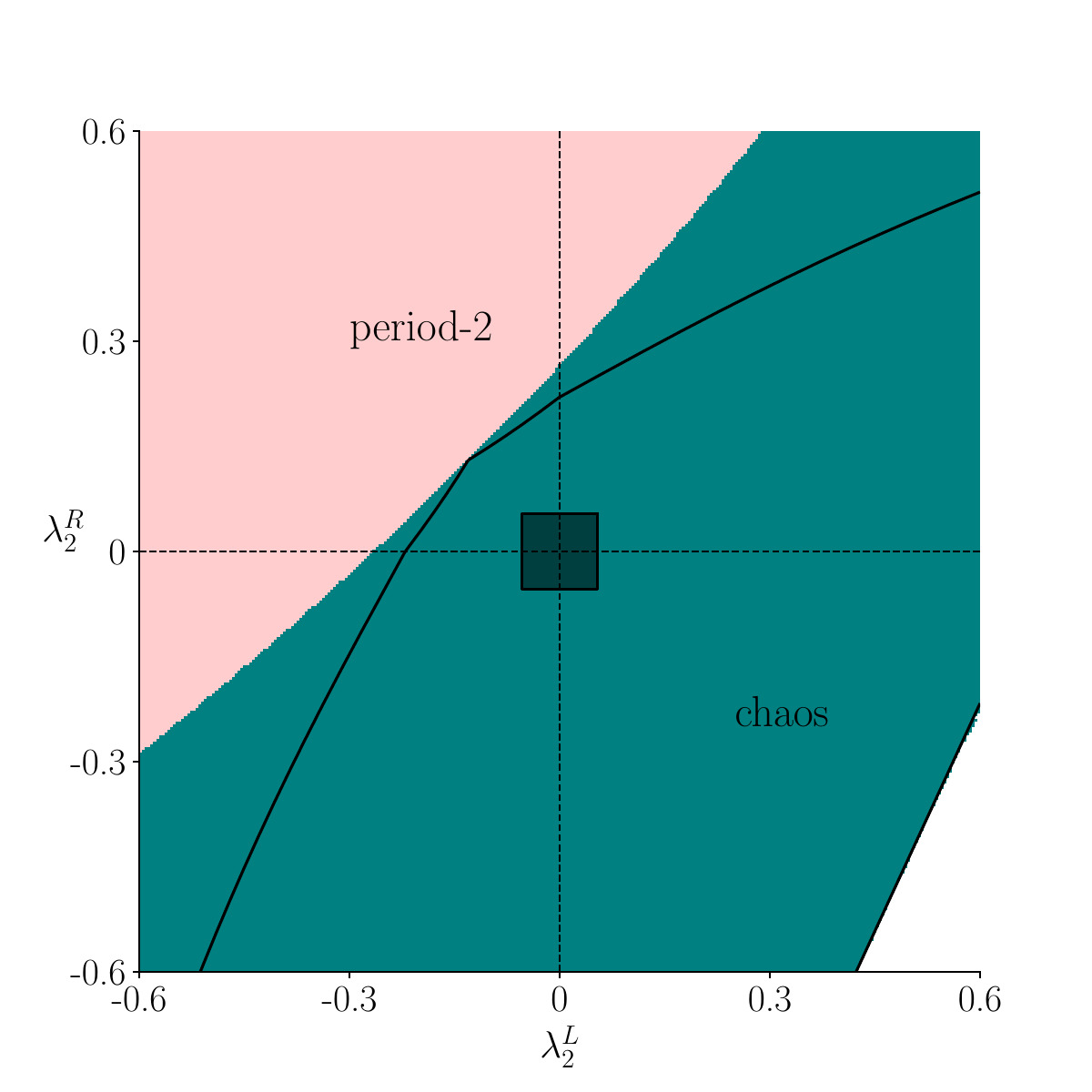}%&
\caption{
A two-parameter bifurcation diagram of the two-dimensional BCNF
with $\alpha = \beta = 1.3$.
The central square is where the conditions of Theorem \ref{th:main} are satisfied,
the region between the black curves are where the results
of earlier work \cite{GhMc23} verifies robust chaos,
while the green region is where numerical simulations suggest
a chaotic attractor exists.
}\label{fig:numerics}
\end{figure}

In Fig.~\ref{fig:numerics} the green region indicates values of 
$\lambda^L_2$ and $\lambda^L_3$ where the numerically computed 
forward orbit of $x = (0,0)$
appeared to converge to an attractor with a positive Lyapunov exponent.
The pink region is where this orbit appeared to converge to a period-two solution,
while the white region is where the orbit appeared to diverge.
%DJWS: the boundary between pink and green is where the period-two solution loses
%stability by attaining an eigenvalue of $-1$.
%This boundary we can compute explicitly, it is the curve
%$$
%\lambda^R_s = \frac{\alpha \beta + \lambda^L_s \alpha + \lambda^R_s \beta - 1}
%{\alpha + \beta + (1-\alpha \beta) \lambda^L_s}.
%$$
%Indra to talk to David regarding this
In our earlier work \cite{GhMc23}
we proved that between the black curves
the map has a chaotic attractor.
As expected this is a subset of the green region where numerics suggests
the presence of a chaotic attractor.
The result of \cite{GhMc23}
was achieved by constructing a trapping region and an invariant expanding cone
using more specialised constructions than what we do below.
In contrast, Theorem \ref{th:main} implies the map has a chaotic attractor
throughout the square in the centre of Fig.~\ref{fig:numerics}.
The sides of this square are where the stable eigenvalues have modulus $r = 0.05$,
which is the minimum value of the right-hand sides
of \eqref{eq:rA}, \eqref{eq:rB}, and \eqref{eq:rC}.
Clearly Theorem \ref{th:main} gives a weaker result than \cite{GhMc23}
for the two-dimensional setting, but it applies to any number
of dimensions and only contains the simple constraints
\eqref{eq:rA}, \eqref{eq:rB}, and \eqref{eq:rC}.
Given the form of these constraints the bound of $r = 0.05$
in the above example cannot be greatly improved upon
due to the proximity of the pink region where the map has a stable period-two
solution.

%===============================================================================
\section{Geometric aspects of discrete dynamics}
\label{sec:basics}

In this section we perform the main constructions.
We define a set $\Omega \subset \mathbb{R}^n$
that over the next two sections we work towards proving is forward invariant,
subject to the assumptions in Theorem \ref{th:main},
and can be perturbed into a trapping region.
We also define a cone $\Psi \subset \mathbb{R}^n$
that in \S\ref{sec:cone} we prove to be contracting-invariant and expanding.

%-------------------------------------------------------------------------------
\subsection{Key definitions}
\label{sub:defns}

We first provide general definitions for the main dynamical tools that are used.
%For further discussion on these can be found in textbooks \cite{Ea98,Ro04}.
Note, ${\rm int}(\cdot)$ denotes {\em interior}.

%...............................................................................
\begin{definition}
For a continuous map $F : \mathbb{R}^n \to \mathbb{R}^n$, a set $S \subset \mathbb{R}^n$ is
\begin{enumerate}
\item
{\em forward invariant} if $F(S) \subset S$, and
\item
a {\em trapping region} if it is non-empty, compact, and $F(S) \subset {\rm int}(S)$.
\end{enumerate}
\label{df:trappingRegion}
\end{definition}

Fig.~\ref{fig:trapping_reg} illustrates the definition of a trapping region.
If $S$ is a trapping region for $F$,
then $\bigcap_{k \ge 0} F^k(S)$ is non-empty and invariant under $F$ \cite{Ea98}.
If this set satisfies an indivisibility property, e.g.~it contains a dense orbit,
then it is, by definition, a topological attractor \cite{Me07,Ro04}.

\begin{figure}[h]
\centering
\includegraphics[width=0.3\linewidth]{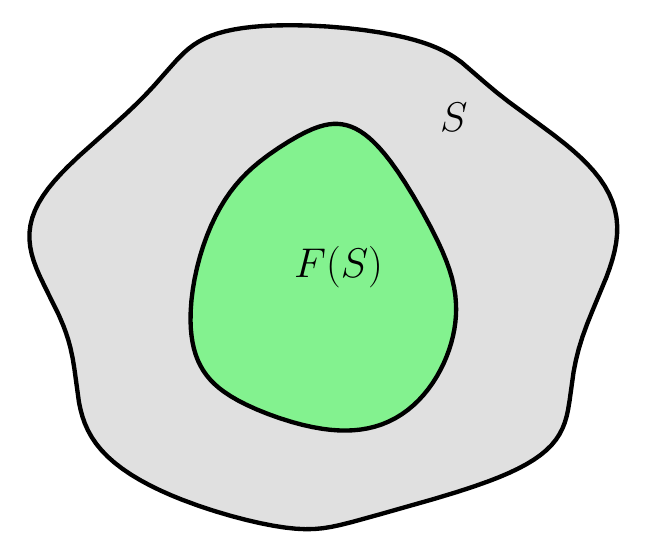}%&
\caption{A trapping region $S$ of a continuous map $F$.}\label{fig:trapping_reg}
\end{figure}

The next definition applies to cones.
A {\em cone} is a non-empty set $\Psi \subset \mathbb{R}^n$
with the property that $t v \in \Psi$ for all $v \in \Psi$ and $t \in \mathbb{R}$.
We use cones because matrix multiplication is linear,
so if any of the three properties in Definition \ref{df:iec} holds for some set $T$,
then it also holds for the smallest cone containing $T$.

%...................................................................................................
\begin{definition}
For a collection $\mathcal{U}$ of real-valued $n \times n$ matrices, a cone $\Psi \subset \mathbb{R}^n$ is
\begin{enumerate}
\item
{\em invariant} if $A v \in \Psi$ for all $A \in \mathcal{U}$ and $v \in \Psi$,
\item
{\em contracting-invariant} if $A v \in {\rm int}(\Psi) \cup \{ \bO \}$ for all $A \in \mathcal{U}$ and $v \in \Psi$, and
\item
{\em expanding} if there exists $c > 1$ such that $\| A v \| \ge c \| v \|$ for all $A \in \mathcal{U}$ and $v \in \Psi$.
\end{enumerate}
\label{df:iec}
\end{definition}

Fig.~\ref{fig:inv_exp_cone} illustrates the definition of an invariant expanding cone.
As shown below, if the collection $\{ A_L, A_R \}$ has an invariant expanding cone
containing a non-zero vector,
then any invariant set of the BCNF has a positive Lyapunov exponent.
The stronger requirement of being contracting-invariant allows
us to accommodate nonlinear terms in the pieces of the map, see \S\ref{sec:bcb}.

\begin{figure}[h]
\centering
\includegraphics[width=0.6\linewidth]{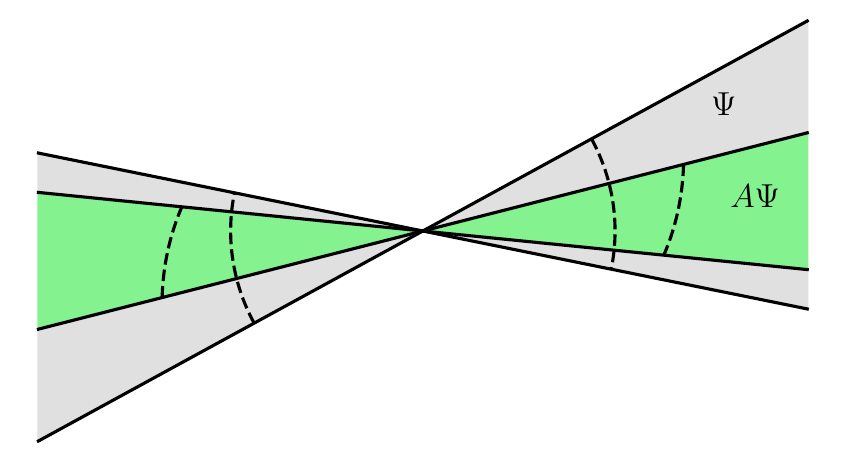}%&
\caption{
A cone $\Psi$ and the set $A \Psi = \left\{ A v \,\middle|\, v \in \Psi \right\}$
for some $A \in \mathbb{R}^{2 \times 2}$.
The dashed curves indicate unit vectors in $\Psi$
and their images under multiplication by $A$.
Here the cone $\Psi$ is continuous and expanding for the set $\mathcal{U} = \{ A \}$.
}\label{fig:inv_exp_cone}
\end{figure}

%-------------------------------------------------------------------------------
\subsection{Trapping region construction}
\label{sub:trappingRegion}

With Assumption \ref{as:main} the fixed point $Y$ has an $(n-1)$-dimensional stable manifold.
As this manifold emanates from $Y$ it coincides with the hyperplane $E$ through $Y$
with directions given by the stable eigenvectors of $A_L$, see already Fig.~\ref{fig:Omega}.
As shown in \S\ref{sub:formulas}, there exists a unique point $C$ on the $x_1$-axis that maps under $f_R$ to $E$.
As shown by Proposition \ref{pr:fCInOmega},
under an additional assumption the point $C$ belongs to the right half-space.
Notice that as $r \to 0$, the points $Y$ and $C$ converge to $(w,0,\ldots,0)$ and $(z,0,\ldots,0)$, respectively,
where $w$ and $z$ are the fixed point and preimage
identified in \S\ref{sub:skewTentMaps}
for the corresponding skew tent map.

For each $i = 2,\ldots,n$, let
\begin{equation}
M_i = \frac{2 \alpha}{\alpha - 1} \begin{pmatrix} n-1 \\ i-1 \end{pmatrix},
\label{eq:Mi}
\end{equation}
and
\begin{equation}
H = \left\{ x \in \mathbb{R}^n \,\middle|\, |x_i| \le M_i r^{i-1} \text{~for all~} i = 2,\ldots,n \right\}.
\label{eq:H}
\end{equation}
Then let
\begin{equation}
K = E \cap H,
\label{eq:K}
\end{equation}
and
\begin{equation}
\Omega = {\rm Conv}(C,K),
\label{eq:Omega}
\end{equation}
where ${\rm Conv}(\cdot)$ denotes {\em convex hull}.
A sketch of $\Omega$ is provided by
Fig.~\ref{fig:Omega} for the three-dimensional setting.

\begin{figure}[h]
\centering
\includegraphics[width=0.6\linewidth]{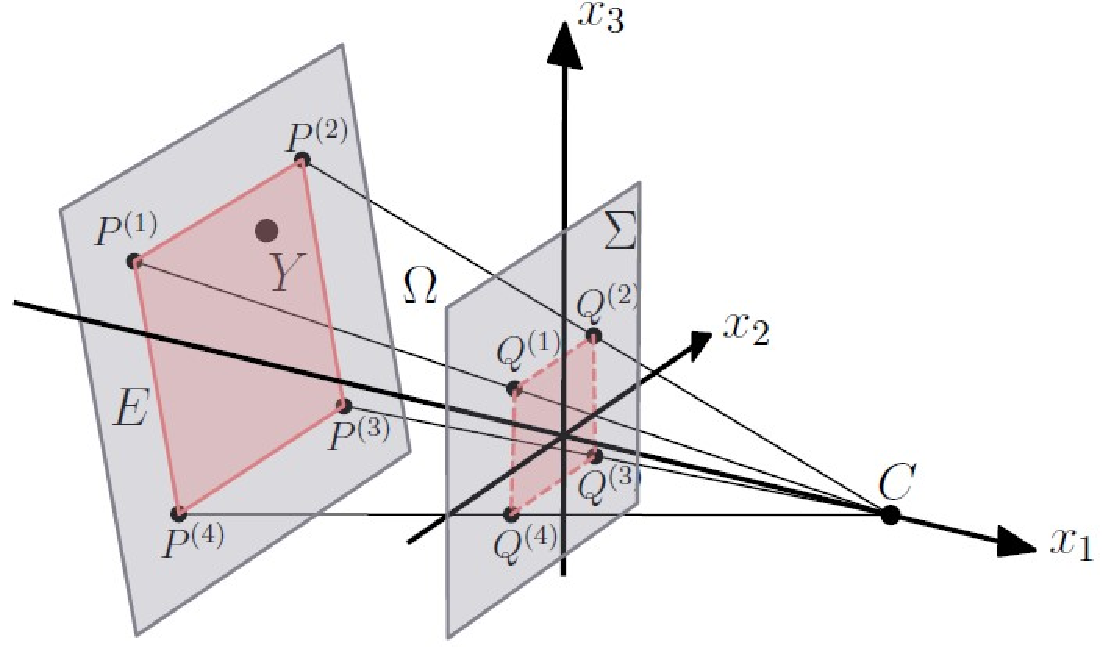}%&
\caption{
The construction of $\Omega$ for $n = 3$.
The hyperplane $E$ is the stable subspace associated with the saddle fixed point $Y$.
The points $P^{(s)}$ for $s = 1,\ldots,4$
belong to $E$ and satisfy $\left| P^{(s)}_2 \right| = M_2 r$
and $\left| P^{(s)}_3 \right| = M_3 r^2$ for each $s$.
The set $K$ is the convex hull of the points $P^{(s)}$,
the point $C$ belongs to the positive $x_1$-axis and maps to $E$,
and $\Omega$ is the convex hull of $K$ and $C$.
}\label{fig:Omega}
\end{figure}

%...................................................................................................
\begin{proposition}
If Assumption \ref{as:main}, \eqref{eq:rA}, and \eqref{eq:rC} hold,
then $\Omega$ is forward invariant under $f$.
Moreover, $\Omega$ can be perturbed into a trapping region.
\label{pr:forwardInvariantRegion}
\end{proposition}

By the last sentence of Proposition \ref{pr:forwardInvariantRegion}
there exists a trapping region $\Omega_\nu$ for $f$,
where $\Omega_\nu$ is a close approximation to $\Omega$.
The set $\Omega$ is not a trapping region itself
because some points in $\Omega$ (e.g.~$C$)
map to the boundary of $\Omega$.

Let us now provide some reasoning behind the above construction.
The set $\Omega$ is polytope (the generalisation of polygon to more than two dimensions),
so $f(\Omega)$ is also polytope because $f$ is piecewise-linear.
Also $\Omega$ is convex, so to prove $f(\Omega) \subset \Omega$
we only need to show that every vertex of $f(\Omega)$ belongs to $\Omega$.

As $\frac{1}{\alpha} + \frac{1}{\beta} \to 1$ and $r \to 0$,
part of the attractor of $f$ approaches $E$ from the right,
so the left boundary of $\Omega$, denoted $K$, cannot be placed a fixed distance to the right of $E$.
Also $K$ cannot contain points located to the left of $E$
because the forward orbits of some such points diverge (while remaining in the
left half-space).
For this reason we have set the left boundary of $\Omega$ to be a subset of $E$.
Below we show $K$ is contained in the left half-space,
see \eqref{eq:PinOmegaProof1}, hence $f(K) \subset E$.
On $K$ the map $f$ draws orbits toward $Y$, which belongs to $K$,
so it is not surprising that $f(K) \subset K$, as needed for forward invariance.

In view of the constraints \eqref{eq:rA} and \eqref{eq:rC}
it is reasonable to treat $r$ as a small quantity.
In this context the parameters $a^L_i$ and $a^R_i$
are order $r^{i-1}$, see Proposition \ref{pr:charPolyCoeffs},
and this is reflected in the definition of $M_i$.
The specific factor $\begin{pmatrix} n-1 \\ i-1 \end{pmatrix}$ in $M_i$
comes naturally from bounds on the values of $a^L_i$, see \eqref{eq:aLi5}.
Also $\Omega$ must contain $Y$, which tends to infinity as $\alpha \to 1$,
so we need $M_i \to \infty$ as $\alpha \to 1$, so we have
included the factor $\alpha - 1$ in the denominator.
The remaining factor $2 \alpha$ in the numerator was chosen through trial-and-error.
With instead significantly smaller values for $M_i$,
the vertices $P^{(s)}$ of $\Omega$ would not map to $\Omega$.
With instead significantly larger values for $M_i$,
the points $Q^{(s)}$
(where the edges of $\Omega$ intersect $\Sigma$) would not map to $\Omega$.

%-------------------------------------------------------------------------------
\subsection{Cone construction}
\label{sub:iec}

For each $i = 1,\ldots,n-1$, let
\begin{equation}
N_i = \frac{\min[\alpha,\beta] - 1}{2} \begin{pmatrix} n-1 \\ i-1 \end{pmatrix},
\label{eq:Ni}
\end{equation}
and define the cone
\begin{equation}
\Psi = \left\{ \gamma v \,\big|\,
\gamma \in \mathbb{R},\, |m_i| \le N_i r^{i-1} \text{~for all~} i = 1,\ldots,n-1 \right\},
\label{eq:Psi}
\end{equation}
where
\begin{equation}
v = \begin{bmatrix} 1 \\ m_1 \\ m_2 \\ \vdots \\ m_{n-1} \end{bmatrix}.
\label{eq:v}
\end{equation}

%...................................................................................................
\begin{proposition}
If Assumption \ref{as:main}, \eqref{eq:rA}, \eqref{eq:rB}, and \eqref{eq:rC} hold,
then $\Psi$ is contracting-invariant and expanding for $\{ A_L, A_R \}$.
\label{pr:cone}
\end{proposition}

Our definition of $\Psi$ is about as simple as possible.
The value of each $m_i$ is constrained to a symmetric interval.
The specific values $N_i$ were chosen for similar reasons to the $M_i$,
but here we need $N_i \to 0$ as $\alpha \to 1$ or $\beta \to 1$
because in these limits the collection of vectors $v$ whose norms increase
under multiplication by $A_L$ and $A_R$ tends to $v_1$-axis.
Smaller values of $N_i$ fail invariance,
while large values of $N_i$ fail expansion.
The factor $\frac{1}{2}$ in the definition of $N_i$
was found to be a reasonably optimal middle value.

%===============================================================================
\section{Preliminary calculations and estimates}
\label{sec:bounds}

Throughout this section we assume Assumption \ref{as:main} is satisfied
and denote the stable eigenvalues of $A_L$ and $A_R$ by
$\lambda^L_i$ and $\lambda^R_i$, respectively, for $i = 2,\ldots,n$.

%-------------------------------------------------------------------------------
\subsection{Formulas for points on $\Omega$}
\label{sub:formulas}

We first derive formulas for the vertices of $\Omega$
and points where its edges meet $\Sigma$.

The fixed point $Y$ is given in terms of $A_L$ and $b$ by \eqref{eq:Y}.
Into \eqref{eq:Y} we substitute the expressions in \eqref{eq:bcnfALARb} for $A_L$ and $b$ to obtain
\begin{equation}
Y = \frac{1}{\Delta}
\begin{bmatrix}
1 \\
-\sum_{j=2}^n a^L_j \\
-\sum_{j=3}^n a^L_j \\
\vdots \\
-a^L_n
\end{bmatrix}
\label{eq:Y2}
\end{equation}
where
\begin{equation}
\Delta = \det(I - A_L) = 1 + \sum_{j=1}^n a^L_j \,.
\label{eq:Delta}
\end{equation}
Notice $\Delta$ is the product of the eigenvalues of $I - A_L$:
\begin{equation}
\Delta = (1-\alpha) \prod_{j=2}^n \left( 1 - \lambda^L_j \right),
\label{eq:Delta2}
\end{equation}
so $\Delta < 0$ because $\alpha > 1$ and $|\lambda^L_j| < 1$ for each $j$.

The row vector
\begin{equation}
u^{\sf T} = \begin{bmatrix}
\alpha^{n-1} & \alpha^{n-2} & \cdots & 1
\end{bmatrix}
\label{eq:u}
\end{equation}
is a left eigenvector of $A_L$ corresponding to the eigenvalue $\alpha$.
This vector is orthogonal to all (right) eigenvectors of $A_L$
corresponding to the stable eigenvalues.
Thus the hyperplane $E$, which contains the initial part of the stable manifold of $Y$,
can be expressed as
\begin{equation}
E = \left\{ x \in \mathbb{R}^n \,\middle|\, u^{\sf T} (x-Y) = 0 \right\}.
\label{eq:E}
\end{equation}
By definition, $C = (C_1,0,\ldots,0)$ satisfies $f_R(C) \in E$.
Thus $u^{\sf T} \left( A_R C + b - Y \right) = 0$, and by solving this equation for $C_1$ we obtain
\begin{equation}
C_1 = \frac{1 - \frac{1}{\Delta} + \frac{1}{\Delta} \sum_{i=2}^n \alpha^{1-i} \sum_{j=i}^n a^L_j}
{\sum_{i=1}^n \alpha^{1-i} a^R_i}.
\label{eq:C1}
\end{equation}
As shown below (Proposition \ref{pr:fCInOmega}),
$C$ belongs to the right half-space
under the assumptions of Proposition \ref{pr:forwardInvariantRegion}.

Let $P^{(s)}$ for $s = 1,2,\ldots,2^{n-1}$ denote the vertices of $\Omega$ on $K$, see Fig.~\ref{fig:Omega}.
For our purposes it is not necessary to specify which vertex each $P^{(s)}$ corresponds to,
only that each $P^{(s)}$ belongs to $E$ and the boundary of $H$.
So $u^{\sf T} \left( P^{(s)} - Y \right) = 0$ and
\begin{equation}
\left| P^{(s)}_i \right| = M_i r^{i-1},
\label{eq:absValuePsi}
\end{equation}
for all $i = 2,\ldots,n$.
By solving $u^{\sf T} \left( P^{(s)} - Y \right) = 0$ for $P^{(s)}_1$, we obtain
\begin{equation}
P^{(s)}_1 = \frac{1}{\Delta} - \sum_{i=2}^n \alpha^{1-i} \left( P^{(s)}_i + \frac{1}{\Delta} \sum_{j=i}^n a^L_j \right).
\label{eq:P1}
\end{equation}
As shown below (see \eqref{eq:PinOmegaProof1}),
under the assumptions of Proposition \ref{pr:forwardInvariantRegion}
each $P^{(s)}_1$ lies in the left half-space.

Finally, for each $s = 1,2,\ldots,2^{n-1}$, let
\begin{equation}
Q^{(s)} = \frac{C_1}{C_1 - P^{(s)}_1}
\begin{bmatrix} 0 \\ P^{(s)}_2 \\ P^{(s)}_3 \\ \vdots \\ P^{(s)}_n \end{bmatrix},
\label{eq:Q}
\end{equation}
be the point where the line through $P^{(s)}$ and $C$
intersects $\Sigma$, see again Fig.~\ref{fig:Omega}.

%-------------------------------------------------------------------------------
\subsection{Characteristic polynomial coefficients}
\label{sub:aZi}

Here we express the parameters of the BCNF in terms of the eigenvalues of $A_L$ and $A_R$.
This is straight-forward because the parameters of the BCNF
are the coefficients of the characteristic polynomials of $A_L$ and $A_R$.
We then use these expressions to obtain bounds on the parameter values
to be used in later proofs.

We first write
\begin{align}
\lambda^L_i &= \eta^L_i r, &
\lambda^R_i &= \eta^R_i r,
\label{eq:etaLi}
\end{align}
for each $i = 2,\ldots,n$.
Notice $\left| \eta^Z_i \right| \le 1$,
for all $i = 2,\ldots,n$ and each $Z \in \{ L, R \}$ by Assumption \ref{as:main}.
For each $i = 1,\ldots,n-1$ and $Z \in \{ L, R \}$, let
\begin{equation}
\xi^Z_i = \text{the sum of all products of $i$ distinct elements of
$\left\{ \eta^Z_2, \eta^Z_3, \ldots, \eta^Z_n \right\}$}.
\label{eq:xiZi}
\end{equation}
For example with $n = 4$,
$\xi^L_1 = \eta^L_2 + \eta^L_3 + \eta^L_4$,
$\xi^L_2 = \eta^L_2 \eta^L_3 + \eta^L_2 \eta^L_4 + \eta^L_3 \eta^L_4$, and
$\xi^L_3 = \eta^L_2 \eta^L_3 \eta^L_4$.

%...................................................................................................
\begin{proposition}
The parameters of \eqref{eq:bcnf} with \eqref{eq:bcnfALARb} satisfy
\begin{equation}
\begin{split}
a^L_1 &= -\alpha - \xi^L_1 r, \\
a^L_i &= (-1)^i \left( \alpha \xi^L_{i-1} + \xi^L_i r \right) r^{i-1}, \qquad \text{for all $i = 2,\ldots,n-1$}, \\
a^L_n &= (-1)^n \alpha \xi^L_{n-1} r^{n-1},
\end{split}
\label{eq:aL}
\end{equation}
and
\begin{equation}
\begin{split}
a^R_1 &= \beta - \xi^R_1 r, \\
a^R_i &= (-1)^i \left( -\beta \xi^R_{i-1} + \xi^R_i r \right) r^{i-1}, \qquad \text{for all $i = 2,\ldots,n-1$}, \\
a^R_n &= (-1)^{n-1} \beta \xi^R_{n-1} r^{n-1}.
\end{split}
\label{eq:aR}
\end{equation}
\label{pr:charPolyCoeffs}
\end{proposition}

Before we prove Proposition \ref{pr:charPolyCoeffs}
we use it to derive formulas for two series that arose in \S\ref{sub:formulas}.
Reindexing \eqref{eq:aL} and \eqref{eq:aR} leads to
\begin{align}
\sum_{j=i}^n a^L_j
&= \sum_{j=i}^{n-1} (-1)^j \left( \alpha \xi^L_{j-1} r^{j-1} + \xi^L_j r^j \right) + (-1)^n \alpha \xi^L_{n-1} r^{n-1} \nonumber \\
&= (-1)^i \alpha \xi^L_{i-1} r^{i-1} - (\alpha-1) \sum_{j=i}^{n-1} (-1)^j \xi^L_j r^j, 
\label{eq:aLjsum}
\end{align}
for any $i = 2,\ldots,n$, and
\begin{align}
\sum_{i=1}^n \alpha^{1-i} a^R_i &= \beta + \sum_{i=2}^n \alpha^{1-i} (-1)^{i-1} \beta \xi^R_{i-1} r^{i-1}
+ \sum_{i=1}^{n-1} \alpha^{1-i} (-1)^i \xi^R_i r^i \nonumber \\
&= \beta + (\alpha + \beta) \sum_{i=1}^{n-1} (-1)^i \xi^R_i \left( \frac{r}{\alpha} \right)^i.
\label{eq:aRisum}
\end{align}

%...................................................................................................
\begin{proof}[Proof of Proposition \ref{pr:charPolyCoeffs}]
We just derive \eqref{eq:aL} as \eqref{eq:aR} is analogous.
We write $\lambda^L_1 = \alpha$ so that the eigenvalues of $A_L$ are $\lambda^L_1,\lambda^L_2,\ldots,\lambda^L_n$ and
the characteristic polynomial of $A_L$ can be written as
\begin{equation}
\det(\lambda I - A_L) = \left( \lambda - \lambda^L_1 \right)
\left( \lambda - \lambda^L_2 \right) \cdots \left( \lambda - \lambda^L_n \right).
\nonumber
\end{equation}
By expanding the brackets and matching the result to \eqref{eq:charPolys} we obtain
\begin{equation}
\begin{split}
a^L_1 &= -\sum_{i=1}^n \lambda^L_i \,, \\
a^L_2 &= \sum_{i=1}^{n-1} \sum_{j=i+1}^n \lambda^L_i \lambda^L_j \,, \\
&\hspace{2.5mm}\vdots \\
a^L_n &= (-1)^n \lambda^L_1 \lambda^L_2 \cdots \lambda^L_n \,.
\end{split}
\nonumber
\end{equation}
Then substituting $\lambda^L_1 = \alpha$ and $\lambda^L_i = \eta^L_i r$ for $i = 2,\ldots,n$ gives
\begin{equation}
\begin{split}
a^L_1 &= -\alpha - \sum_{i=2}^n \eta^L_i r, \\
a^L_2 &= \alpha \sum_{i=2}^n \eta^L_i r + \sum_{i=2}^{n-1} \sum_{j=i+1}^n \eta^L_i \eta^L_j r^2, \\
&\hspace{2.5mm}\vdots \\
a^L_n &= (-1)^n  \alpha \eta^L_2 \cdots \eta^L_n r^{n-1}.
\end{split}
\nonumber
\end{equation}
Finally by inserting the definition of $\xi^L_i$ we obtain \eqref{eq:aL}.
\end{proof}

%...................................................................................................
\begin{corollary}
The parameters of the BCNF satisfy
\begin{align}
\left| a^L_1 + \alpha \right| &\le r(n-1), \label{eq:aL15} \\
\left| a^L_i \right| &\le \begin{pmatrix} n-1 \\ i-1 \end{pmatrix} \left( \alpha + \frac{(n-i) r}{i} \right) r^{i-1},
\qquad \text{for all $i = 2,\ldots,n-1$}, \label{eq:aLi5} \\
\left| a^L_n \right| &\le \alpha r^{n-1}, \label{eq:aLn5}
\end{align}
and
\begin{align}
\left| a^R_1 - \beta \right| &\le r(n-1), \label{eq:aR15} \\
\left| a^R_i \right| &\le \begin{pmatrix} n-1 \\ i-1 \end{pmatrix} \left( \beta + \frac{(n-i) r}{i} \right) r^{i-1},
\qquad \text{for all $i = 2,\ldots,n-1$}, \label{eq:aRi5} \\
\left| a^R_n \right| &\le \beta r^{n-1}. \label{eq:aRn5}
\end{align}
\label{co:aZiBounds}
\end{corollary}

%...................................................................................................
\begin{proof}
Since $\left| \eta^Z_i \right| \le 1$ for all $i = 2,\ldots,n$,
the absolute value of each sum $\xi^Z_i$ is bounded by the number of terms in the sum.
That is,
\begin{equation}
\left| \xi^Z_i \right| \le \begin{pmatrix} n-1 \\ i \end{pmatrix},
\label{eq:xiZiBound}
\end{equation}
because the terms come from choosing $i$ values from a set with $n-1$ elements.
The result then follows by applying \eqref{eq:xiZiBound} to \eqref{eq:aL} and \eqref{eq:aR}.
\end{proof}

%-------------------------------------------------------------------------------
\subsection{Polynomial to linear bounds}

Since $\left| \lambda^L_j \right| \le r$ for all $j = 2,\ldots,n$,
the formula \eqref{eq:Delta2} for the determinant $\Delta$ of $I - A_L$ immediately gives
\begin{equation}
-(\alpha-1)(1+r)^{n-1} \le \Delta \le -(\alpha-1)(1-r)^{n-1}.
\label{eq:DeltaBound}
\end{equation}
These bounds are polynomial functions of $r$.
Here we derive bounds that are weaker but linear functions of $r$,
so easier to work with in later sections.

%...................................................................................................
\begin{proposition}
If $0 < \ee < 1$, $n \ge 2$, and $0 \le r(n-1) \le \ee$, then
\begin{equation}
|\Delta + \alpha - 1| \le \frac{\alpha-1}{1-\ee} \,r(n-1),
\label{eq:DeltaBound2}
\end{equation}
and
\begin{equation}
\left| \frac{1}{\Delta} + \frac{1}{\alpha - 1} \right| \le \frac{1}{(\alpha-1)(1-\ee)} \,r(n-1).
\label{eq:DeltaBound4}
\end{equation}
\label{pr:DeltaBounds}
\end{proposition}

To prove Proposition \ref{pr:DeltaBounds} we use the following estimates, illustrated in Fig.~\ref{fig:linearBounds}.

%...................................................................................................
\begin{lemma}
If $0 < \ee < 1$, $n \ge 2$, and $0 \le r(n-1) \le \ee$, then
\begin{align}
(1+r)^{n-1} &\le 1 + \frac{r (n-1)}{1-\ee}, \label{eq:bound1} \\
(1+r)^{-(n-1)} &\ge 1 - r (n-1), \label{eq:bound2} \\
(1-r)^{n-1} &\ge 1 - r (n-1), \label{eq:bound3} \\
(1-r)^{-(n-1)} &\le 1 + \frac{r (n-1)}{1-\ee}. \label{eq:bound4}
\end{align}
\label{le:linearBounds}
\end{lemma}

\begin{figure}[h]
\centering
\begin{tabular}{cc}
  \includegraphics[scale=.5]{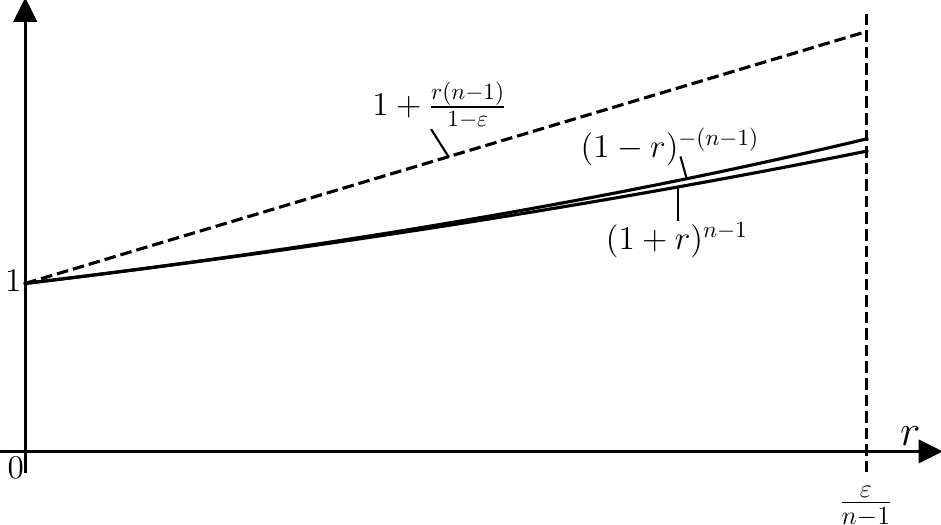} &   \includegraphics[scale=0.5]{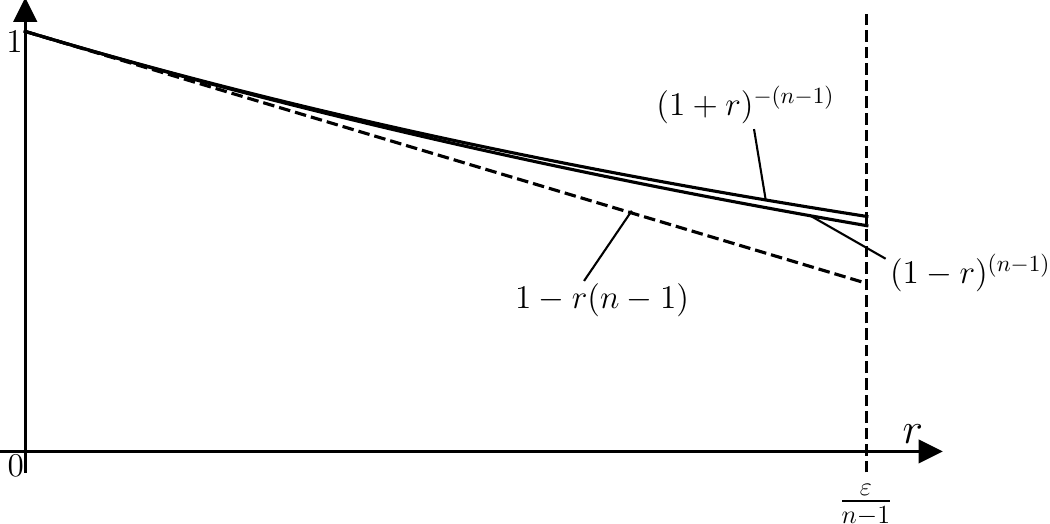} \\
(a)  & (b)  \\[6pt]
\end{tabular}
\caption{Linear bounds represented by the estimates~\eqref{eq:bound1}--\eqref{eq:bound4}.}\label{fig:linearBounds}
\end{figure}

%...................................................................................................
\begin{proof}
We just derive \eqref{eq:bound3} and \eqref{eq:bound4}
as \eqref{eq:bound1} and \eqref{eq:bound2} can be obtained similarly.
%DJWS: for \eqref{eq:bound1} we take $n \to \infty$ instead of $n = 2$
%and the result follows from observing ${\rm e}^\ee \le \frac{1}{1-\ee}$
%which can be seen by comparing the Taylor series of each side centred at $\ee = 0$.
As a function of $r$, the left-hand side of \eqref{eq:bound3} is concave up,
so is bounded from below by its tangent line at $r=0$, which is the right-hand side of \eqref{eq:bound3}.
Now observe $F(r,n) = \frac{(1-r)^{-(n-1)} - 1}{r(n-1)}$ is an increasing function of $r$,
thus $F(r,n) \le F \left( \frac{\ee}{n-1}, n \right)$ for all $0 \le r(n-1) \le \ee$.
But $F \left( \frac{\ee}{n-1}, n \right)$ is a decreasing function of $n$,
so $F(r,n) \le F \left( \ee, 2 \right)$ because $n \ge 2$.
That is,
\begin{equation}
\frac{(1-r)^{-(n-1)} - 1}{r(n-1)} \le \frac{(1-\ee)^{-\ee} - 1}{\ee} = \frac{1}{1-\ee},
\nonumber
\end{equation}
which verifies \eqref{eq:bound4}.
\end{proof}

%...................................................................................................
\begin{proof}[Proof of Proposition \ref{pr:DeltaBounds}]
By applying \eqref{eq:bound1} and \eqref{eq:bound3} to \eqref{eq:DeltaBound} we obtain
\begin{equation}
-(\alpha-1) \left( 1 + \tfrac{r(n-1)}{1-\ee} \right) \le \Delta \le -(\alpha-1)(1 - r(n-1)),
\label{eq:DeltaBound1}
\end{equation}
and hence \eqref{eq:DeltaBound2}.
By inverting \eqref{eq:DeltaBound} and then applying \eqref{eq:bound2} and \eqref{eq:bound4} we obtain
\begin{equation}
\frac{-1}{\alpha-1} \left( 1 + \tfrac{r(n-1)}{1-\ee} \right) \le \frac{1}{\Delta} \le \frac{-1}{\alpha-1}(1 - r(n-1)),
\nonumber
\end{equation}
and hence \eqref{eq:DeltaBound4}.
\end{proof}

%-------------------------------------------------------------------------------
\subsection{Bounds on series}

Here we derive bounds on two series that appeared in \S\ref{sub:formulas}.
To do this we use the binomial theorem, specifically
\begin{equation}
\sum_{k=0}^{n-1} \begin{pmatrix} n-1 \\ k \end{pmatrix} t^k = (1 + t)^{n-1},
\label{eq:binomialTheorem}
\end{equation}
for any $t \in \mathbb{R}$.

%...................................................................................................
\begin{proposition}
If $0 < \ee < 1$, $n \ge 2$, and $0 \le r(n-1) \le \ee$, then
\begin{equation}
\frac{1}{|\Delta|} \sum_{i=2}^n \alpha^{1-i} \left| \sum_{j=i}^n a^L_j \right| \le
\frac{(1+\ee)}{(\alpha-1)(1-\ee)} \,r(n-1),
\label{eq:aLdoubleSumBound}
\end{equation}
and
\begin{equation}
\left| \frac{1}{\sum_{i=1}^n \alpha^{1-i} a^R} - \frac{1}{\beta} \right|
\le \frac{2}{\beta (1 - 3 \ee)} \,r(n-1).
\label{eq:aRisumBound}
\end{equation}
\label{pr:sumBounds}
\end{proposition}

%...................................................................................................
\begin{proof}
By applying \eqref{eq:xiZiBound} to \eqref{eq:aLjsum},
\begin{equation}
\left| \sum_{j=i}^n a^L_j \right| \le \alpha \begin{pmatrix} n-1 \\ i-1 \end{pmatrix} r^{i-1}
+ (\alpha-1) \sum_{j=i}^{n-1} \begin{pmatrix} n-1 \\ j \end{pmatrix} r^j.
\label{eq:aLjsum2}
\end{equation}
By \eqref{eq:DeltaBound1} and $r(n-1) \le \ee$,
\begin{equation}
|\Delta| \ge (\alpha-1)(1 - \ee),
\nonumber
\end{equation}
and so by \eqref{eq:aLjsum2},
\begin{equation}
\frac{1}{|\Delta|} \sum_{i=2}^n \alpha^{1-i} \left| \sum_{j=i}^n a^L_j \right| \le
\frac{\alpha}{(\alpha-1)(1-\ee)} 
\sum_{i=2}^n \begin{pmatrix} n-1 \\ i-1 \end{pmatrix} \left( \frac{r}{\alpha} \right)^{i-1}
+ \frac{1}{1-\ee} \sum_{i=2}^n \alpha^{1-i} \sum_{j=i}^{n-1} \begin{pmatrix} n-1 \\ j \end{pmatrix} r^j.
\label{eq:sumBoundsZProof2}
\end{equation}
By the binomial theorem
\begin{equation}
\sum_{i=2}^n \begin{pmatrix} n-1 \\ i-1 \end{pmatrix} \left( \frac{r}{\alpha} \right)^{i-1}
= -1 + \left( 1 + \frac{r}{\alpha} \right)^{n-1},
\nonumber
\end{equation}
and so by \eqref{eq:bound1} (with $\frac{r}{\alpha}$ in place of $r$),
\begin{equation}
\sum_{i=2}^n \begin{pmatrix} n-1 \\ i-1 \end{pmatrix} \left( \frac{r}{\alpha} \right)^{i-1}
\le \frac{r (n-1)}{\alpha (1-\ee)}.
\label{eq:sumBoundsZProof10}
\end{equation}
Next we bound the double series
\begin{equation}
D = \sum_{i=2}^n \alpha^{1-i} \sum_{j=i}^{n-1} \begin{pmatrix} n-1 \\ j \end{pmatrix} r^j,
\nonumber
\end{equation}
that appears in \eqref{eq:sumBoundsZProof2}.
Switching the order of summation gives
\begin{equation}
D = \sum_{j=2}^{n-1} \begin{pmatrix} n-1 \\ j \end{pmatrix} r^j \sum_{i=2}^j \alpha^{1-i}.
\nonumber
\end{equation}
Then by $\sum_{i=2}^j \alpha^{1-i} < \sum_{i=2}^\infty \alpha^{1-i} = \frac{1}{\alpha - 1}$,
for the value of a geometric series,
\begin{equation}
D \le \frac{1}{\alpha - 1} \sum_{j=2}^{n-1} \begin{pmatrix} n-1 \\ j \end{pmatrix} r^j.
\nonumber
\end{equation}
By the binomial theorem
\begin{equation}
D \le \frac{1}{\alpha - 1} \left( -1 - r(n-1) + (1+r)^{n-1} \right),
\nonumber
\end{equation}
and then by \eqref{eq:bound1}
\begin{equation}
D \le \frac{\ee r (n-1)}{(\alpha-1)(1-\ee)}.
\label{eq:sumBoundsZProof20}
\end{equation}
The desired bound \eqref{eq:aLdoubleSumBound} results
from substituting \eqref{eq:sumBoundsZProof10} and \eqref{eq:sumBoundsZProof20} into \eqref{eq:sumBoundsZProof2}.

Next we derive \eqref{eq:aRisumBound}.
By \eqref{eq:aRisum},
\begin{equation}
\frac{1}{\sum_{i=1}^n \alpha^{1-i} a^R} - \frac{1}{\beta}
= \frac{1}{\beta} \left( -1 + \frac{1}{1 + \left( 1 + \frac{\alpha}{\beta} \right)
\sum_{i=1}^{n-1} (-1)^i \xi^R_i \left( \frac{r}{\alpha} \right)^i} \right).
\label{eq:aRisumBoundProof1}
\end{equation}
By \eqref{eq:xiZiBound}, the binomial theorem, and \eqref{eq:bound1} (with $\frac{r}{\alpha}$ in place of $r$),
\begin{align}
\left| \sum_{i=1}^{n-1} (-1)^i \xi^R_i \left( \frac{r}{\alpha} \right)^i \right|
\le \sum_{i=1}^{n-1} \begin{pmatrix} n-1 \\ i \end{pmatrix} \left( \frac{r}{\alpha} \right)^i
= -1 + \left( 1 + \frac{r}{\alpha} \right)^{n-1}
\le \frac{r(n-1)}{\alpha (1-\ee)}.
\nonumber
\end{align}
Substituting this into \eqref{eq:aRisumBoundProof1} gives
\begin{align}
\left| \frac{1}{\sum_{i=1}^n \alpha^{1-i} a^R} - \frac{1}{\beta} \right|
\le \frac{1}{\beta} \left( -1 + \frac{1}{1 - \left( 1 + \frac{\alpha}{\beta} \right)
\frac{r(n-1)}{\alpha (1-\ee)}} \right)
= \frac{1}{\beta} \frac{r(n-1)}{\frac{\alpha \beta (1-\ee)}{\alpha + \beta} - r(n-1)}.
\label{eq:aRisumBoundProof3}
\end{align}
Finally, $\alpha > 1$ and $\beta > 1$ imply $\frac{1}{\alpha} + \frac{1}{\beta} < 2$,
so $\frac{\alpha \beta}{\alpha + \beta} > \frac{1}{2}$.
By inserting this and $r(n-1) \le \ee$ into the denominator of
the last expression in \eqref{eq:aRisumBoundProof3}
we arrive at \eqref{eq:aRisumBound}.
\end{proof}

%===============================================================================
\section{Calculations for the polytope $\Omega$}
\label{sec:forwardInvariantRegion}

In this section we prove Proposition \ref{pr:forwardInvariantRegion}. % that $f(\Omega) \subset \Omega$.
In \S\ref{sub:Ps} we show each $f \left( P^{(s)} \right)$ belongs to $\Omega$,
in \S\ref{sub:C} we show $f(C)$ belongs to $\Omega$,
and in \S\ref{sub:Qs} we show each $f \left( Q^{(s)} \right)$ belongs to $\Omega$.
Proposition \ref{pr:forwardInvariantRegion} then follows from the convexity of $\Omega$
and the linearity of each piece of $f$, see \S\ref{sub:OmegaFinal}.

Several results and proofs involve a small quantity $\ee > 0$.
For the purposes of proving Proposition \ref{pr:forwardInvariantRegion}
is sufficient to consider $\ee = \frac{1}{10}$ (the coefficient in \eqref{eq:rC}),
but in many places we leave $\ee$ unspecified to make the calculations more transparent.

%-------------------------------------------------------------------------------
\subsection{Each $P^{(s)}$ maps into $\Omega$}
\label{sub:Ps}

%...................................................................................................
\begin{proposition}
If Assumption \ref{as:main} holds and $r(n-1) \le \frac{1}{10}$,
then $P^{(s)}_1 < 0$ and $f \left( P^{(s)} \right) \in E \cap {\rm int}(H)$ for all $s = 1,2,\ldots,2^{n-1}$.
\label{pr:fPsInOmega}
\end{proposition}

To prove Proposition \ref{pr:fPsInOmega}
we first derive bounds on the first component of $P^{(s)}$.

%...................................................................................................
\begin{lemma}
If Assumption \ref{as:main} holds and $r(n-1) \le \ee < 1$, then
\begin{equation}
\left| P^{(s)}_1 + \frac{1}{\alpha - 1} \right| \le \frac{2(2-\ee)}{(\alpha-1)(1-\ee)^2} \,r(n-1),
\label{eq:P1Bound}
\end{equation}
for all $s = 1,2,\ldots,2^{n-1}$.
\label{le:P1Bound}
\end{lemma}

%...................................................................................................
\begin{proof}
The first component of $P^{(s)}$ is given by \eqref{eq:P1}, so
\begin{equation}
\left| P^{(s)}_1 + \frac{1}{\alpha - 1} \right| \le
\left| \frac{1}{\Delta} + \frac{1}{\alpha - 1} \right| +
\sum_{i=2}^n \alpha^{1-i} \left( \left| P^{(s)}_i \right| + \frac{1}{|\Delta|} \left| \sum_{j=i}^n a^L_j \right| \right),
\label{eq:P1BoundProof1}
\end{equation}
using the triangle inequality.
By \eqref{eq:Mi} and \eqref{eq:absValuePsi},
\begin{equation}
\sum_{i=2}^n \alpha^{1-i} \left| P^{(s)}_i \right|
= \frac{2 \alpha}{1-\alpha} \sum_{i=2}^n \begin{pmatrix} n-1 \\ i-1 \end{pmatrix} \left( \frac{r}{\alpha} \right)^{i-1},
\nonumber
\end{equation}
so by \eqref{eq:sumBoundsZProof10} we have
\begin{equation}
\sum_{i=2}^n \alpha^{1-i} \left| P^{(s)}_i \right|
\le \frac{2 r (n-1)}{(1-\alpha)(1-\ee)}.
\label{eq:P1BoundProof2}
\end{equation}
By substituting \eqref{eq:DeltaBound4}, \eqref{eq:aLdoubleSumBound}, and \eqref{eq:P1BoundProof2}
into \eqref{eq:P1BoundProof1} we obtain \eqref{eq:P1Bound}.
\end{proof}

%...................................................................................................
\begin{proof}[Proof of Proposition \ref{pr:fPsInOmega}]
Fix $\ee = \frac{1}{10}$.
Since $r(n-1) \le \frac{1}{10}$ and $\frac{2(2-\ee)}{(1-\ee)^2} = \frac{380}{81} < 10$,
the right-hand side of \eqref{eq:P1Bound} is less than $\frac{1}{\alpha-1}$.
Thus
\begin{equation}
-\frac{2}{\alpha - 1} < P^{(s)}_1 < 0.
\label{eq:PinOmegaProof1}
\end{equation}
so $P^{(s)}$ maps under $f_L$ as claimed.
Moreover, $f \left( P^{(s)} \right) \in E$ because $P^{(s)} \in E$
and $E$ is invariant under $f_L$.

It remains to show $f \left( P^{(s)} \right) \in {\rm int}(H)$,
that is $\left| f \left( P^{(s)} \right) \right|_i < M_i r^{i-1}$, for all $i = 2,\ldots,n$.
We have
\begin{equation}
f \left( P^{(s)} \right) = \begin{bmatrix}
a^L_1 P^{(s)}_1 + P^{(s)}_2 + 1 \\
a^L_2 P^{(s)}_1 + P^{(s)}_3 \\
\vdots \\
a^L_{n-1} P^{(s)}_1 + P^{(s)}_n \\
a^L_n P^{(s)}_1
\end{bmatrix}.
\label{eq:fP}
\end{equation}
In particular,
\begin{equation}
\left| f \left( P^{(s)} \right)_n \right| = \left| a^L_n \right| \left| P^{(s)}_1 \right|,
\nonumber
\end{equation}
and by \eqref{eq:aLn5} and \eqref{eq:PinOmegaProof1},
\begin{equation}
\left| f \left( P^{(s)} \right) \right|_n < \alpha r^{n-1} \times \frac{2}{\alpha-1} = M_n r^{n-1}.
\nonumber
\end{equation}
Now let $i = 2,\ldots,n-1$.
By \eqref{eq:fP},
\begin{equation}
\left| f \left( P^{(s)} \right)_i \right| \le \left| a^L_i \right| \left| P^{(s)}_1 \right| + \left| P^{(s)}_{i+1} \right|.
\nonumber
\end{equation}
Inserting \eqref{eq:aLi5}, \eqref{eq:P1Bound}, and $\left| P^{(s)}_{i+1} \right| = M_{i+1} r^i$ gives
\begin{equation}
\left| f \left( P^{(s)} \right)_i \right| \le \frac{\alpha}{\alpha - 1} \begin{pmatrix} n-1 \\ i-1 \end{pmatrix} r^{i-1}
\left( \left( 1 + \frac{(n-i) r}{\alpha i} \right) \left( 1 + \frac{2(2-\ee)}{(1-\ee)^2} \,r(n-1) \right) + \frac{2 (n-i) r}{i} \right),
\nonumber
\end{equation}
after factoring.
Next we substitute $\frac{n-i}{i} < \frac{n-1}{2}$, $\alpha > 1$, and $r(n-1) \le \ee$ resulting in
\begin{equation}
\left| f \left( P^{(s)} \right)_i \right|
\le \frac{\kappa_1 \alpha}{\alpha - 1} \begin{pmatrix} n-1 \\ i-1 \end{pmatrix} r^{i-1}
= \frac{\kappa_1 M_i}{2} \,r^{i-1},
\nonumber
\end{equation}
where
\begin{equation}
\kappa_1 = \left( 1 + \frac{\ee}{2} \right) \left( 1 + \frac{2 \ee (2 - \ee)}{(1-\ee)^2} \right) + \ee.
\nonumber
\end{equation}
But $\ee = \frac{1}{10}$, so $\kappa_1 = \frac{887}{540} < 2$,
thus $\left| f \left( P^{(s)} \right)_i \right| < M_i r^{i-1}$ as required.
\end{proof}

%-------------------------------------------------------------------------------
\subsection{The vertex $C$ maps into $\Omega$}
\label{sub:C}

%...................................................................................................
\begin{proposition}
If Assumption \ref{as:main} holds and $r(n-1) \le \frac{1}{10}$,
then $C_1 > 0$ and $f(C) \in E \cap {\rm int}(H)$.
\label{pr:fCInOmega}
\end{proposition}

To prove Proposition \ref{pr:fCInOmega} we use tight bounds on the value of $C_1$:

%...................................................................................................
\begin{lemma}
If Assumption \ref{as:main} holds and $r(n-1) \le \ee < 1$, then
\begin{equation}
\left| C_1 - \frac{\alpha}{(\alpha-1) \beta} \right| \le \frac{2 \alpha + 2 + \ee}{(1-3\ee) (\alpha - 1) \beta} \,r(n-1).
\label{eq:C1Bound}
\end{equation}
\label{le:C1Bounds}
\end{lemma}

%...................................................................................................
\begin{proof}
The formula \eqref{eq:C1} is
\begin{equation}
C_1 = \frac{1}{T} \left( 1 - \frac{1}{\Delta} + S \right),
\label{eq:C12}
\end{equation}
where we now write $S = \frac{1}{\Delta} \sum_{i=2}^n \alpha^{1-i} \sum_{j=i}^n a^L_j$
and $T = \sum_{i=1}^n \alpha^{1-i} a^R_i$.
Equation \eqref{eq:C12} can be rearranged as
\begin{equation}
C_1 - \frac{\alpha}{(\alpha-1) \beta} = \frac{\alpha}{\alpha - 1} \left( \frac{1}{T} - \frac{1}{\beta} \right)
- \frac{1}{T} \left( \frac{1}{\Delta} + \frac{1}{\alpha - 1} \right) + \frac{S}{T},
\label{eq:C1BoundsProof10}
\end{equation}
so by the triangle inequality
\begin{equation}
\left| C_1 - \frac{\alpha}{(\alpha-1) \beta} \right| \le
\frac{\alpha}{\alpha - 1} \left| \frac{1}{T} - \frac{1}{\beta} \right|
+ \frac{1}{|T|} \left( \left| \frac{1}{\Delta} + \frac{1}{\alpha - 1} \right| + |S| \right).
\label{eq:C1BoundsProof11}
\end{equation}
Putting $r(n-1) \le \ee$ into \eqref{eq:aRisumBound} gives
\begin{equation}
\frac{1}{|T|} \le \frac{1}{\beta} \left( 1 + \frac{2 \ee}{1 - 3 \ee} \right) = \frac{1-\ee}{\beta (1 - 3 \ee)}.
\label{eq:aRisumBoundWeak}
\end{equation}
By then substituting \eqref{eq:DeltaBound4}, \eqref{eq:aLdoubleSumBound},
\eqref{eq:aRisumBound}, and \eqref{eq:aRisumBoundWeak}
into \eqref{eq:C1BoundsProof11} we obtain
\begin{align}
\left| C_1 - \frac{\alpha}{(\alpha-1) \beta} \right|
&\le \frac{\alpha}{\alpha - 1} \,\frac{2 r(n-1)}{\beta (1-3\ee)} + \frac{1-\ee}{\beta (1 - 3 \ee)}
\left( \frac{r(n-1)}{(\alpha-1)(1-\ee)} + \frac{(1+\ee) r(n-1)}{(\alpha-1)(1-\ee)} \right) \nonumber \\
&= \frac{2 \alpha + 2 + \ee}{(1-3\ee)(\alpha - 1) \beta} \,r(n-1),
\label{eq:C1BoundsProof20}
\end{align}
as required.
\end{proof}

%...................................................................................................
\begin{proof}[Proof of Proposition \ref{pr:fCInOmega}]
Fix $\ee = \frac{1}{10}$; this allows us to apply Lemma \ref{le:C1Bounds}.
To achieve some simplification, in the right-hand side of \eqref{eq:C1Bound}
we replace the numerator $2 \alpha + 2 + \ee$ with the larger value $(4 + \ee) \alpha$
to produce
\begin{equation}
\left| C_1 - \frac{\alpha}{(\alpha-1) \beta} \right| \le \frac{(4 + \ee) \alpha}{(1-3\ee) (\alpha - 1) \beta} \,r(n-1).
\label{eq:C1BoundWeaker}
\end{equation}
This implies
\begin{equation}
C_1 \ge \frac{\alpha}{(\alpha - 1) \beta} \left( 1 - \frac{\ee (4 + \ee)}{1 - 3 \ee} \right),
\nonumber
\end{equation}
and notice $1 - \frac{\ee (4 + \ee)}{1 - 3 \ee} = \frac{29}{70} > 0$, so $C_1 > 0$ as claimed.
Thus $f(C) = f_R(C)$, hence $f(C) \in E$ by the definition of $C$.

It remains to show $f(C) \in {\rm int}(H)$,
that is, $|f(C)_i| < M_i r^{i-1}$ for all $i = 2,\ldots,n$.
Since $C_1 > 0$ is the only non-zero component of $C$,
\begin{equation}
f(C) = \begin{bmatrix} a^R_1 C_1 + 1 \\ a^R_2 C_1 \\ \vdots \\ a^R_n C_1 \end{bmatrix}.
\label{eq:fC}
\end{equation}
In particular,
\begin{equation}
|f(C)_n| = \left| a^R_n \right| |C_1|,
\nonumber
\end{equation}
and by \eqref{eq:aRn5} and \eqref{eq:C1BoundWeaker},
\begin{align}
|f(C)_n| &\le \beta r^{n-1} \left( \frac{\alpha}{(\alpha-1) \beta} + \frac{(4 + \ee) \alpha}{(1-3 \ee)(\alpha-1) \beta} \,r(n-1) \right)
\nonumber \\
&= \frac{\alpha}{\alpha - 1} \left( 1 + \frac{4+\ee}{1-3 \ee} \,r(n-1) \right) r^{n-1}. \nonumber 
\end{align}
By using $r(n-1) \le \ee$ this reduces to
\begin{equation}
|f(C)_n| \le \frac{\kappa_2 \alpha}{\alpha - 1} \,r^{n-1}
= \frac{\kappa_2 M_n}{2} \,r^{n-1},
\nonumber
\end{equation}
where $\kappa_2 = \frac{1 + \ee + \ee^2}{1 - 3 \ee}$.
But $\ee = \frac{1}{10}$, so $\kappa_2 = \frac{111}{70} < 2$,
thus $|f(C)_n| < M_n r^{n-1}$.

Now let $i = 2,\ldots,n-1$.
By \eqref{eq:fC},
\begin{equation}
|f(C)_i| = \left| a^R_i \right| |C_1|,
\nonumber
\end{equation}
and by \eqref{eq:aRi5} and \eqref{eq:C1BoundWeaker},
\begin{align}
|f(C)_i| &\le \begin{pmatrix} n-1 \\ i-1 \end{pmatrix} \left( \beta + \frac{(n-i) r}{i} \right) r^{i-1}
\left( \frac{\alpha}{(\alpha-1) \beta} + \frac{(4 + \ee) \alpha}{(1-3 \ee)(\alpha-1) \beta} \,r(n-1) \right) \nonumber \\
&\le \frac{\alpha}{\alpha-1} \begin{pmatrix} n-1 \\ i-1 \end{pmatrix}
\left( 1 + \frac{r (n-1)}{2 \beta} \right)
\left( 1 + \frac{4+\ee}{1-3 \ee} \,r(n-1) \right) r^{i-1},
\nonumber
\end{align}
using also $\frac{(n-i) r}{i} < \frac{r (n-1)}{2}$ to produce the second line.
By then using $\beta > 1$ and $r(n-1) \le \ee$ this reduces to
\begin{equation}
|f(C)_i| \le \frac{\kappa_3 \alpha}{\alpha - 1} \begin{pmatrix} n-1 \\ i-1 \end{pmatrix} r^{i-1}
= \frac{\kappa_3 M_i}{2} \,r^{i-1},
\nonumber
\end{equation}
where $\kappa_3 = \left( 1 + \frac{\ee}{2} \right) \kappa_2$.
But $\ee = \frac{1}{10}$, so $\kappa_3 = \frac{333}{200} < 2$,
thus $|f(C)_i| < M_i r^{i-1}$ as required.
\end{proof}

%-------------------------------------------------------------------------------
\subsection{Each $Q^{(s)}$ maps into $\Omega$}
\label{sub:Qs}

%...................................................................................................
\begin{proposition}
If Assumption \ref{as:main}, \eqref{eq:rA}, and \eqref{eq:rC} hold,
then $f \left( Q^{(s)} \right) \in {\rm int}(\Omega)$ for all $s = 1,2,\ldots,2^{n-1}$.
\label{pr:fQsInOmega}
\end{proposition}

Proposition \ref{pr:fQsInOmega} is proved below after we establish four lemmas.
First Lemma \ref{le:gQs1} shows that each $f \left( Q^{(s)} \right)$
belongs to the right half-space,
so certainly lies to the right of the boundary $K$ of $\Omega$.
To show $f \left( Q^{(s)} \right)$ lies to the left of $C$,
we use the formula \eqref{eq:Q} for $Q^{(s)}$ to obtain
\begin{align}
C_1 - f \left( Q^{(s)} \right)_1 = C_1 - \left( \frac{C_1 P^{(s)}_2}{C_1 - P^{(s)}_1} + 1 \right)
= C_1 \left( 1 - \frac{1}{C_1} - \frac{P^{(s)}_2}{C_1 - P^{(s)}_1} \right).
\label{eq:gQsinOmegaProof1}
\end{align}
Lemmas \ref{le:QsPart1} and \ref{le:QsPart2}
provide bounds on $1 - \frac{1}{C_1}$ and $\frac{P^{(s)}_2}{C_1 - P^{(s)}_1}$, respectively.
These show $C_1 - f \left( Q^{(s)} \right)_1 > 0$,
i.e.~$f \left( Q^{(s)} \right)$ lies to the left of $C$.
But more importantly they show $C_1 - f \left( Q^{(s)} \right)_1$ is large enough
that we can then show $f \left( Q^{(s)} \right)$ lies within the remaining boundaries of $\Omega$.
Lemma \ref{le:QsPart3} provides an additional bound that enables us to complete the proof.

%...................................................................................................
\begin{lemma}
If Assumption \ref{as:main} and \eqref{eq:rA} hold and $r(n-1) \le \frac{1}{10}$,
then $f \left( Q^{(s)} \right)_1 > 0$ for all $s = 1,2,\ldots,2^{n-1}$.
\label{le:gQs1}
\end{lemma}

%...................................................................................................
\begin{proof}
By \eqref{eq:Q},
\begin{equation}
f \left( Q^{(s)} \right)_1 = \frac{C_1 P^{(s)}_2}{C_1 - P^{(s)}_1} + 1.
%\label{eq:gQs1Proof1}
\nonumber
\end{equation}
Observe $0 < \frac{C_1}{C_1 - P^{(s)}_1} < 1$, because
$P^{(s)}_1 < 0$ (Proposition \ref{pr:fPsInOmega}) and $C_1 > 0$ (Proposition \ref{pr:fCInOmega}).
Also
\begin{equation}
\left| P^{(s)}_2 \right| = M_2 r = \frac{2 \alpha r(n-1)}{\alpha - 1},
\label{eq:gQs1Proof2}
\end{equation}
so
\begin{equation}
f \left( Q^{(s)} \right)_1 > 1 - \frac{2 \alpha r(n-1)}{\alpha - 1}.
\nonumber
\end{equation}
Then substituting \eqref{eq:rA} gives
\begin{equation}
f \left( Q^{(s)} \right)_1 > 1 - \frac{6}{7} > 0.
\nonumber
\end{equation}
\end{proof}

%...................................................................................................
\begin{lemma}
If Assumption \ref{as:main} and \eqref{eq:rC} hold, then
\begin{equation}
1 - \frac{1}{C_1} > \frac{2 \beta \phi}{3},
\label{eq:Qspart1}
\end{equation}
where
\begin{equation}
\phi = \frac{1}{\alpha} + \frac{1}{\beta} - 1.
\label{eq:phi}
\end{equation}
\label{le:QsPart1}
\end{lemma}

%...................................................................................................
\begin{proof}
Fix $\ee = \frac{1}{10}$.
By \eqref{eq:C1Bound}
\begin{equation}
C_1 \ge \frac{\alpha}{(\alpha-1) \beta} \left( 1 - \frac{2 \alpha + 2 + \ee}{(1 - 3 \ee) \alpha} \,r(n-1) \right).
\label{eq:Qspart1Proof20}
\end{equation}
We now perform a series of calculations that lead towards a bound on the right-hand side of \eqref{eq:Qspart1Proof20}.
First observe
\begin{equation}
\alpha^2 - \frac{23}{10} \,\alpha + \frac{21}{5} > 0,
\nonumber
\end{equation}
because $\alpha > 1$.
Since $\ee = \frac{1}{10}$ this is
\begin{equation}
\left( \frac{1}{\ee} - 9 \right) \alpha^2 - (2 + 3 \ee) \alpha + 2(2+\ee) > 0,
\label{eq:Qspart1Proof50}
\end{equation}
which can be rearranged to form
\begin{equation}
3 - \frac{2}{\alpha} < \frac{(1 - 3 \ee) \alpha}{\ee (2 \alpha + 2 + \ee)}.
\nonumber
\end{equation}
Since $\beta > 1$ this implies
\begin{equation}
2 + \frac{1}{\beta} - \frac{2}{\alpha} < \frac{(1 - 3 \ee) \alpha}{\ee (2 \alpha + 2 + \ee)},
\nonumber
\end{equation}
which can be rearranged as
\begin{equation}
\ee < \frac{(1 - 3 \ee) \alpha \beta}{(2 \alpha + 2 + \ee)(3 - 2 \beta \phi)}.
\label{eq:Qspart1Proof10}
\end{equation}
The bound \eqref{eq:rC} is $r(n-1) \le \ee \phi$, so
\begin{equation}
r(n-1) < \frac{(1-3 \ee) \alpha \beta \phi}{(2 \alpha + 2 + \ee)(3 - 2 \beta \phi)},
\nonumber
\end{equation}
which can be rearranged as
\begin{equation}
\frac{\alpha}{(\alpha - 1) \beta} \left( 1 - \frac{2 \alpha + 2 + \ee}{(1 - 3 \ee) \alpha} \,r(n-1) \right)
> \frac{3}{3 - 2 \beta \phi}.
\nonumber
\end{equation}
So by \eqref{eq:Qspart1Proof20} we have
\begin{equation}
C_1 > \frac{3}{3 - 2 \beta \phi},
\nonumber
\end{equation}
which rearranges to \eqref{eq:Qspart1}.
\end{proof}

%...................................................................................................
\begin{lemma}
If Assumption \ref{as:main} and \eqref{eq:rC} hold, then
\begin{equation}
\frac{\left| P^{(s)}_2 \right|}{C_1 - P^{(s)}_1} < \frac{\beta \phi}{3},
\label{eq:Qspart2}
\end{equation}
for all $s = 1,2,\ldots,2^{n-1}$ and
where $\phi$ is given by \eqref{eq:phi}.
\label{le:QsPart2}
\end{lemma}

%...................................................................................................
\begin{proof}
Fix $\ee = \frac{1}{10}$.
By \eqref{eq:P1Bound},
\begin{equation}
P^{(s)}_1 \le \frac{1}{\alpha - 1} \left( 1 - \frac{2(2-\ee)}{(1-\ee)^2} \,r(n-1) \right),
\label{eq:Qspart2Proof1}
\end{equation}
while $P^{(s)}_2$ and $C_1$ are bounded by \eqref{eq:gQs1Proof2} and \eqref{eq:Qspart1Proof20}, respectively.
By substituting these bounds into the left-hand side of \eqref{eq:Qspart2},
and also $r(n-1) \le \ee$ in \eqref{eq:Qspart1Proof20} and \eqref{eq:Qspart2Proof1}, we obtain
\begin{equation}
\frac{\left| P^{(s)}_2 \right|}{C_1 - P^{(s)}_1} \le
\frac{2 \alpha}{(\alpha - 1) U} \,r(n-1),
\label{eq:Qspart2Proof10}
\end{equation}
where
\begin{equation}
U = \frac{\alpha}{(\alpha-1) \beta} \left( 1 - \frac{(2 \alpha + 2 + \ee) \ee}{(1 - 3 \ee) \alpha} \right)
+ \frac{1}{\alpha - 1} \left( 1 - \frac{2 \ee (2-\ee)}{(1-\ee)^2} \right).
\label{eq:Qspart2Proof12}
\end{equation}
Observe
\begin{equation}
\frac{4}{35} \alpha + \frac{187}{810} \beta + \frac{3}{10} (\beta - 1) > 0,
\nonumber
\end{equation}
which can be rearranged to
\begin{equation}
\frac{\alpha}{10} < \frac{1}{6} \left( \frac{5 \alpha}{7} + \frac{43 \beta}{81} - \frac{3}{10} \right).
\nonumber
\end{equation}
Since $\ee = \frac{1}{10}$, this is equivalent to
\begin{align}
\ee \alpha < \frac{1}{6} \left( \left( 1 - \frac{2 \ee}{1 - 3 \ee} \right) \alpha
+ \left( 1 - \frac{2 \ee (2 - \ee)}{(1-\ee)^2} \right) \beta
- \frac{\ee (2 + \ee)}{1 - 3 \ee} \right)
= \frac{(\alpha-1) \beta U}{6}.
\nonumber
\end{align}
But \eqref{eq:rC} is $r(n-1) \le \ee \phi$, so
\begin{equation}
r(n-1) < \frac{(\alpha-1) \beta \phi U}{6 \alpha},
\nonumber
\end{equation}
and with \eqref{eq:Qspart2Proof10} this gives \eqref{eq:Qspart2}.
\end{proof}

%...................................................................................................
\begin{lemma}
If Assumption \ref{as:main} holds and $r(n-1) \le \frac{1}{10}$, then
\begin{equation}
\frac{C_1 - P^{(s)}_1}{C_1 - P^{(t)}_1} > \frac{3}{20 \beta},
\label{eq:Qspart3}
\end{equation}
for all $s,t = 1,2,\ldots,2^{n-1}$.
\label{le:QsPart3}
\end{lemma}

%...................................................................................................
\begin{proof}
Fix $\ee = \frac{1}{10}$.
Using \eqref{eq:P1Bound}, \eqref{eq:C1Bound}, and $r(n-1) \le \ee$,
\begin{equation}
\frac{C_1 - P^{(s)}_1}{C_1 - P^{(t)}_1} \ge \frac{U}{V}
\nonumber
\end{equation}
where $U$ is given by \eqref{eq:Qspart2Proof12} and
\begin{equation}
V = \frac{\alpha}{(\alpha-1) \beta} \left( 1 + \frac{(2 \alpha + 2 + \ee) \ee}{(1 - 3 \ee) \alpha} \right)
+ \frac{1}{\alpha - 1} \left( 1 + \frac{2 \ee (2-\ee)}{(1-\ee)^2} \right).
\nonumber
\end{equation}
Since $\ee = \frac{1}{10}$,
\begin{equation}
\frac{U}{V} - \frac{3 \ee}{2 \beta} = \frac{\frac{2}{945} \alpha \beta + \frac{86}{81} \beta^2 + \frac{27}{70} \alpha (\beta - 1)
+ \frac{281}{270} (\alpha-1) \beta + \frac{9}{100}}
{2 (\alpha - 1) \beta^2 V},
\nonumber
\end{equation}
which, by inspection, is positive, verifying \eqref{eq:Qspart3}.
\end{proof}

%...................................................................................................
\begin{proof}[Proof of Proposition \ref{pr:fQsInOmega}]
Choose any $s = 1,2,\ldots,2^{n-1}$.
The point $f \left( Q^{(s)} \right)$ lies to the right of $K$
because $K$ is a subset of the left half-space (Proposition \ref{pr:fPsInOmega})
while $f \left( Q^{(s)} \right)$ lies in the right half-space (Lemma \ref{le:gQs1}).
By \eqref{eq:gQsinOmegaProof1}, \eqref{eq:Qspart1}, and \eqref{eq:Qspart2},
\begin{equation}
C_1 - f \left( Q^{(s)} \right)_1 > C_1 \left( \frac{2 \beta \phi}{3} - \frac{\beta \phi}{3} \right)
= \frac{\beta \phi C_1}{3}
\label{eq:gQsinOmegaProof2}
\end{equation}
is positive, thus $f \left( Q^{(s)} \right)$ lies to the left of $C$.
To show $f \left( Q^{(s)} \right) \in \Omega$,
let $t = 1,2,\ldots,2^{n-1}$ and $W^{(s,t)}$
be the point on the line through $P^{(t)}$ and $C$
whose first component is the same as the first component of $f \left( Q^{(s)} \right)$.
It remains for us to show
\begin{equation}
\left| f \left( Q^{(s)} \right)_i \right| < \left| W^{(s,t)}_i \right|,
\label{eq:gQsinOmegaProof20}
\end{equation}
for all $t = 1,2,\ldots,2^{n-1}$ and all $i = 2,\ldots,n$.

By \eqref{eq:Q},
\begin{equation}
f \left( Q^{(s)} \right) = \begin{bmatrix}
\frac{C_1 P^{(s)}_2}{C_1 - P^{(s)}_1} + 1 \\
\frac{C_1 P^{(s)}_3}{C_1 - P^{(s)}_1} \\
\vdots \\
\frac{C_1 P^{(s)}_n}{C_1 - P^{(s)}_1} \\
0
\end{bmatrix}.
\label{eq:fQ}
\end{equation}
In particular, $f \left( Q^{(s)} \right)_n = 0$,
so \eqref{eq:gQsinOmegaProof20} is true with $i = n$
because $P^{(t)}_n \ne 0$ and hence $W^{(s,t)}_n \ne 0$.

Choose any $t = 1,2,\ldots,2^{n-1}$ and $i = 2,\ldots,n-1$.
The point $W^{(s,t)}$ is defined by $W^{(s,t)} = \theta P^{(t)} + (1-\theta) C$,
where $\theta$ is such that $W^{(s,t)}_1 = f \left( Q^{(s)} \right)_1$,
so $\theta = \frac{C_1 - f \left( Q^{(s)} \right)_1}{C_1 - P^{(t)}_1}$.
Also $W^{(s,t)}_i = \theta P^{(t)}_i$, because $C_i = 0$, so by also using \eqref{eq:fQ} we obtain
\begin{equation}
\left| W^{(s,t)}_i \right| - \left| f \left( Q^{(s)} \right)_i \right|
= \frac{C_1 - f \left( Q^{(s)} \right)_1}{C_1 - P^{(t)}_1} \left| P^{(t)}_i \right|
- \frac{C_1}{C_1 - P^{(s)}_1} \left| P^{(s)}_{i+1} \right|.
\nonumber
\end{equation}
Then by \eqref{eq:Mi} and \eqref{eq:absValuePsi},
\begin{align}
\left| W^{(s,t)}_i \right| - \left| f \left( Q^{(s)} \right)_i \right| &= \frac{2 \alpha}{\alpha - 1} \left(
\frac{C_1 - f \left( Q^{(s)} \right)_1}{C_1 - P^{(t)}_1} \begin{pmatrix} n-1 \\ i-1 \end{pmatrix} r^{i-1}
- \frac{C_1}{C_1 - P^{(s)}_1} \begin{pmatrix} n-1 \\ i \end{pmatrix} r^i \right) \nonumber \\
&> \frac{2 \alpha}{\alpha - 1} \begin{pmatrix} n-1 \\ i-1 \end{pmatrix} \left(
\frac{C_1 - f \left( Q^{(s)} \right)_1}{C_1 - P^{(t)}_1}
- \frac{C_1}{2 \big( C_1 - P^{(s)}_1 \big)} \,r(n-1) \right) r^{i-1},
\end{align}
where we have also used
\begin{equation}
\begin{pmatrix} n-1 \\ i \end{pmatrix}
= \frac{n-i}{i} \begin{pmatrix} n-1 \\ i-1 \end{pmatrix}
< \frac{n-1}{2} \begin{pmatrix} n-1 \\ i-1 \end{pmatrix},
\nonumber
\end{equation}
because $i \ge 2$.
Then by \eqref{eq:Qspart3} and \eqref{eq:gQsinOmegaProof2}
\begin{equation}
\left| W^{(s,t)}_i \right| - \left| f \left( Q^{(s)} \right)_i \right|
> \frac{2 \alpha}{(\alpha - 1) \big( C_1 - P^{(s)}_1 \big)} \begin{pmatrix} n-1 \\ i-1 \end{pmatrix} \left(
\frac{3 \ee}{2 \beta} \times \frac{\beta \phi C_1}{3}
- \frac{C_1}{2} \,r(n-1) \right) r^{i-1},
\label{eq:gQsinOmegaProof40}
\end{equation}
where $\ee = \frac{1}{10}$.
But \eqref{eq:rC} is $r(n-1) \le \ee \phi$, 
so the right-hand side of \eqref{eq:gQsinOmegaProof40} is positive,
verifying \eqref{eq:gQsinOmegaProof20}.
\end{proof}

%-------------------------------------------------------------------------------
\subsection{Final arguments}
\label{sub:OmegaFinal}

%...............................................................................
\begin{proof}[Proof of Proposition \ref{pr:forwardInvariantRegion}]
Let $\Omega_L$ and $\Omega_R$ be the parts of $\Omega$ in $x_1 \le 0$ and $x_1 \ge 0$ respectively.
These are polytopes; the vertices of $\Omega_L$ are $P^{(s)}$ and $Q^{(s)}$ for $s = 1,2,\ldots,2^{n-1}$,
while the vertices of $\Omega_R$ are $C$ and $Q^{(s)}$ for $s = 1,2,\ldots,2^{n-1}$.
The set $\Omega_L$ maps under $f_L$, which is affine,
so $f(\Omega_L)$ is a polytope with vertices given by the images of each $P^{(s)}$ and $Q^{(s)}$ under $f_L$.
Similarly $\Omega_R$ maps under $f_R$,
so $f(\Omega_R)$ is a polytope with vertices given by the images of $C$ and each $Q^{(s)}$ under $f_R$.
By Propositions \ref{pr:fPsInOmega}, \ref{pr:fCInOmega}, and \ref{pr:fQsInOmega},
the images of $C$ and each $P^{(s)}$ and $Q^{(s)}$ belong to $\Omega$.
Thus $f(\Omega_L)$ and $f(\Omega_R)$ are subsets of $\Omega$ because $\Omega$ is convex.
Finally, $f(\Omega) = f(\Omega_L) \cup f(\Omega_R)$, thus $\Omega$ is forward invariant under $f$.

We now perturb $\Omega$ to form a new polytope $\Omega_\nu$
by decreasing the first component of $C$ by a small amount $\nu > 0$,
increasing the first components of each $P^{(s)}$ by $\nu^2$,
and defining $\Omega_\nu$ to be the convex hull of the perturbed points.
It remains for us to show that $\Omega_\nu$ is a trapping region for $f$ if $\nu > 0$ is sufficiently small.
Refer to \cite{GhMc23} for calculations verifying this in a more explicit manner
for a similar construction in the two-dimensional setting.

Each $Q^{(s)}$ maps to the interior of $\Omega$ (Proposition \ref{pr:fQsInOmega}),
so the analogous intersection points of $\Omega_\nu$ with $\Sigma$
map to the interior of $\Omega_\nu$ for sufficiently small $\nu > 0$.
Let $K_\nu$ denote the left boundary of $\Omega_\nu$.
This face is parallel to $K$ and is an order-$\nu^2$ perturbation of $K$.
The point $C$ maps to $K$ and inside all other boundaries of $\Omega$ (Proposition \ref{pr:fCInOmega}).
Consequently the corresponding vertex of $\Omega_\nu$ maps to the right of $K$
and thus to the right of $K_\nu$, because this boundary has been perturbed by a far smaller amount,
so to the interior of $\Omega_\nu$ if $\nu > 0$ is sufficiently small.
Finally, each $P^{(s)}$ maps to $K$ and inside all other boundaries of $\Omega$ (Proposition \ref{pr:fPsInOmega}).
Thus the corresponding vertices of $\Omega_\nu$ map to the right of $K_\nu$, due to the saddle nature of $Y$,
and hence to the interior of $\Omega_\nu$ if $\nu > 0$ is sufficiently small.
Since $\Omega_\nu$ is convex and $f_L$ and $f_R$ are affine,
the entire set $\Omega_\nu$ maps under $f$ to its interior.
\end{proof}

%===============================================================================
\section{Calculations for the cone $\Psi$}
\label{sec:cone}

We now perform calculations for the cone $\Psi$.
In \S\ref{sub:contractingInvariant} we prove contracting-invariance;
in \S\ref{sub:expanding} we prove expansion.
This verifies Proposition \ref{pr:cone},
then in \S\ref{sub:PsiFinal} we prove Theorem \ref{th:main}.

%-------------------------------------------------------------------------------
\subsection{Contracting-invariance}
\label{sub:contractingInvariant}

%...................................................................................................
\begin{proposition}
If Assumption \ref{as:main}, \eqref{eq:rA}, and \eqref{eq:rB} hold, and $r(n-1) \le \frac{1}{10}$,
then $\Psi$ is contracting-invariant for $\{ A_L, A_R \}$.
\label{pr:coneInvariantL}
\end{proposition}

%...................................................................................................
\begin{proof}
Choose any $v \in \Psi$.
Our task is to show $A_Z v \in {\rm int}(\Psi) \cup \{ \bO \}$ for each $Z \in \{ L, R \}$.
By linearity it suffices to assume $v_1 = 1$.
Then we can write $v$ as \eqref{eq:v}, where
\begin{equation}
|m_i| \le N_i r^{i-1}, \text{~and~} N_i = \frac{\min[\alpha,\beta] - 1}{2} \begin{pmatrix} n-1 \\ i-1 \end{pmatrix},
\label{eq:NiAgain}
\end{equation}
for all $i = 1,\ldots,n-1$.
Multiplication by $A_Z$ gives
\begin{equation}
A_Z v = \begin{bmatrix} -a^Z_1 + m_1 \\ -a^Z_2 + m_2 \\ \vdots \\ -a^Z_{n-1} + m_{n-1} \\ -a^Z_n \end{bmatrix}.
\label{eq:coneInvariantLProof1}
\end{equation}
This is a scalar multiple of
\begin{equation}
\begin{bmatrix} 1 \\ m_1' \\ m_2' \\ \vdots \\ m_{n-1}' \end{bmatrix},
\nonumber
\end{equation}
where
\begin{align}
m_i' &= \frac{-a^Z_{i+1} + m_{i+1}}{-a^Z_1 + m_1}, \qquad \text{for all $i = 1,\ldots,n-2$}, \label{eq:miPrime} \\
m_{n-1}' &= \frac{-a^Z_n}{-a^Z_1 + m_1}. \label{eq:mnm1Prime}
\end{align}
It remains for us show $\left| m_i' \right| < N_i r^{i-1}$ for all $i = 1,\ldots,n-1$.

For the remainder of the proof we just treat the case $Z = L$
because the calculations required for $Z = R$ are the same except the roles of $\alpha$ and $\beta$ are reversed.
This is because $N_i$ and the combination of \eqref{eq:rA} and \eqref{eq:rB}
are unchanged when $\alpha$ and $\beta$ are interchanged,
and because the bounds on $a^R_i$ obtained from \eqref{eq:aR15}--\eqref{eq:aRn5}
are the same as those on $a^L_i$ obtained from \eqref{eq:aL15}--\eqref{eq:aLn5}
with $\alpha$ and $\beta$ interchanged.

We first weaken \eqref{eq:aLi5} slightly to simply the ensuing algebra.
With $i \ge 2$ we have $\frac{n-i}{i} < \frac{n-1}{2}$, so
\begin{equation}
\left| a^L_i \right| < \begin{pmatrix} n-1 \\ i-1 \end{pmatrix} \left( \alpha + \frac{r(n-1)}{2} \right) r^{i-1},
\qquad \text{for all $i = 2,\ldots,n-1$}.
\label{eq:aLi6}
\end{equation}
Now choose any $i = 1,\ldots,n-2$ (the case $i = n-1$ is done at the end).
Into \eqref{eq:miPrime} we substitute \eqref{eq:aL15}, \eqref{eq:NiAgain}, and \eqref{eq:aLi6} to produce
\begin{align}
\left| m_i' \right| <
\frac{\begin{pmatrix} n-1 \\ i \end{pmatrix} \left( \alpha + \frac{r(n-1)}{2} \right) r^i
+ \dfrac{\min[\alpha,\beta] - 1}{2} \begin{pmatrix} n-1 \\ i \end{pmatrix} r^i}
{\alpha - r(n-1) - \dfrac{\min[\alpha,\beta] - 1}{2}}.
\nonumber
\end{align}
By also substituting $r(n-1) \le \frac{1}{10}$ this becomes
\begin{align}
\left| m_i' \right| < J \begin{pmatrix} n-1 \\ i \end{pmatrix} r^i,
\label{eq:coneInvariantLProof21}
\end{align}
where
\begin{equation}
J = \frac{2 \alpha + \min[\alpha,\beta] - \frac{9}{10}}{2 \alpha - \min[\alpha,\beta] + \frac{4}{5}}.
\label{eq:J}
\end{equation}
In Appendix \ref{app:proofs} (Lemma \ref{le:boundForContractingInv}) we show 
\begin{equation}
\frac{3}{7} < \frac{\min[\alpha,\beta]}{2 J}.
\label{eq:boundForContractingInv}
\end{equation}
But $r(n-1) < \frac{3}{7} \left( 1 - \frac{1}{\min[\alpha,\beta]} \right)$, by \eqref{eq:rA} and \eqref{eq:rB},
so $r(n-1) < \frac{\min[\alpha,\beta] - 1}{2 J}$.
That is,
\begin{equation}
J < \frac{\min[\alpha,\beta] - 1}{2 r(n-1)},
\label{eq:rAB}
\end{equation}
and so \eqref{eq:coneInvariantLProof21} reduces to
\begin{equation}
\left| m_i' \right| < \frac{\min[\alpha,\beta] - 1}{2 r (n-1)} \begin{pmatrix} n-1 \\ i \end{pmatrix} r^i.
\label{eq:coneInvariantLProof31}
\end{equation}
Into \eqref{eq:coneInvariantLProof31} we insert
\begin{equation}
\begin{pmatrix} n-1 \\ i \end{pmatrix}
= \frac{n-i}{i} \begin{pmatrix} n-1 \\ i-1 \end{pmatrix}
< (n-1) \begin{pmatrix} n-1 \\ i-1 \end{pmatrix},
\label{eq:ChoosyBound}
\end{equation}
to obtain the desired bound
\begin{equation}
\left| m_i' \right| < \frac{\min[\alpha,\beta] - 1}{2} \begin{pmatrix} n-1 \\ i-1 \end{pmatrix} r^{i-1}
= N_i r^{i-1}.
\nonumber
\end{equation}

Finally we treat $i=n-1$.
Into \eqref{eq:mnm1Prime} we substitute \eqref{eq:aL15}, \eqref{eq:aLn5},
and $|m_1| \le \frac{\alpha-1}{2}$ to produce
\begin{equation}
\left| m_{n-1}' \right| < \frac{\alpha r^{n-1}}{\alpha - r(n-1) - \frac{\min[\alpha,\beta] - 1}{2}}
\le \frac{2 \alpha r^{n-1}}{2 \alpha - \min[\alpha,\beta] + \frac{4}{5}},
\nonumber
\end{equation}
using also $r(n-1) \le \frac{1}{10}$.
Then by \eqref{eq:rAB},
\begin{equation}
\left| m_{n-1}' \right|
< \frac{\min[\alpha,\beta] - 1}{2 r(n-1)} \times \frac{2 \alpha r^{n-1}}{2 \alpha + \min[\alpha,\beta] - \frac{9}{10}},
\nonumber
\end{equation}
which implies $\left| m_{n-1}' \right| < N_{n-1} r^{n-2}$, as required,
because $\frac{1}{(n-1)^2} \le 1$
and $\frac{2 \alpha}{2 \alpha + \min[\alpha,\beta] - \frac{9}{10}} < 1$.
\end{proof}

%-------------------------------------------------------------------------------
\subsection{Expansion}
\label{sub:expanding}

%...................................................................................................
\begin{proposition}
If Assumption \ref{as:main}, \eqref{eq:rA}, \eqref{eq:rB}, and \eqref{eq:rC} hold,
then $\Psi$ is expanding for $\{ A_L, A_R \}$.
\label{pr:coneExpandingL}
\end{proposition}

%...................................................................................................
\begin{lemma}
If Assumption \ref{as:main} and \eqref{eq:rA} hold and $r(n-1) \le \frac{1}{10 \alpha}$, then
\begin{equation}
r(n-1) \le \frac{\alpha - 1}{3 \alpha - 1 + \tfrac{1}{10} (\alpha - 1) \left( 2 \alpha + \tfrac{1}{10} \right)}.
\label{eq:HProofEnd}
\end{equation}
\label{le:HProofEnd}
\end{lemma}

%...................................................................................................
\begin{proof}
If $\alpha \le \frac{5}{4}$ then
$303 - 143 \alpha - 60 \alpha^2 > 0$.
This is equivalent to
\begin{equation}
\frac{3}{7} \left( 1 - \frac{1}{\alpha} \right)
< \frac{\alpha - 1}{3 \alpha - 1 + \tfrac{1}{10} (\alpha - 1) \left( 2 \alpha + \tfrac{1}{10} \right)},
\nonumber
\end{equation}
so \eqref{eq:HProofEnd} is a consequence of \eqref{eq:rA}.

If instead $\alpha > \frac{5}{4}$ then
$101 - 1281 \alpha + 980 \alpha^2 > 0$.
This is equivalent to
\begin{equation}
\frac{1}{10 \alpha} < \frac{\alpha - 1}{3 \alpha - 1 + \tfrac{1}{10} (\alpha - 1) \left( 2 \alpha + \tfrac{1}{10} \right)}
\label{eq:}
\end{equation}
so again \eqref{eq:HProofEnd} holds.
\end{proof}

%...................................................................................................
\begin{proof}[Proof of Proposition \ref{pr:coneExpandingL}]
Choose any $v \in \Psi$ with $v_1 = 1$ and let
\begin{equation}
H = \| A_Z v \|^2 - \| v \|^2,
\nonumber
\end{equation}
where $Z \in \{ L, R \}$.
It remains to show $H > 0$.
This is because for each $Z \in \{ L, R \}$ the function $H$ is continuous over the compact set
$\Psi_1 = \left\{ v \in \Psi \,\middle|\, v_1 = 1 \right\}$,
so has a minimum value $h_Z > 0$.
That is, $\| A_Z v \|^2 - \| v \|^2 \ge h_Z$,
which rearranges to $\| A_Z v \| \ge \sqrt{\frac{h_Z}{\| v \|^2} + 1} \,\| v \|$.
So $\Psi$ is expanding for $\{ A_L, A_R \}$ with
expansion factor $c = \sqrt{\frac{\min[h_L,h_R]}{\xi^2} + 1}$,
where $\xi = \max_{v \in \Psi_1} \| v \|$.

Write $v$ in the form \eqref{eq:v} and observe
\begin{align}
\| v \|^2 &= 1 + \sum_{i=1}^{n-1} m_i^2 \,, \nonumber \\
\| A_Z v \|^2 &= \sum_{i=1}^n \left( a^Z_i \right)^2 - 2 \sum_{i=1}^{n-1} a^Z_i m_i + \sum_{i=1}^{n-1} m_i^2 \,. \nonumber
\end{align}
Thus
\begin{equation}
H = \sum_{i=1}^n \left( a^Z_i \right)^2 - 2 \sum_{i=1}^{n-1} a^Z_i m_i - 1,
\label{eq:H1}
\end{equation}
To the first series in \eqref{eq:H1} we retain only the $i=1$ term,
while to the second series we bring the $i=1$ term out the front:
\begin{equation}
H \ge \left( a^Z_1 \right)^2 - 2 |a^Z_1| |m_1| - 2 \sum_{i=2}^{n-1} |a^Z_i| |m_i| - 1.
\label{eq:H2}
\end{equation}

For the remainder of the proof we just treat $Z = L$ as the calculations for $Z = R$
are identical except with $\beta$ instead of $\alpha$.
Into \eqref{eq:H2} with $Z = L$ we substitute \eqref{eq:aL15}, \eqref{eq:NiAgain}, \eqref{eq:aLi6},
and $\min[\alpha,\beta] \le \alpha$, to produce
\begin{align}
H &\ge \left( \alpha - r(n-1) \right)^2 - 2 \left( \alpha + r(n-1) \right) \times \frac{\alpha - 1}{2} \nonumber \\
&\quad- 2 \sum_{i=2}^{n-1} \begin{pmatrix} n-1 \\ i-1 \end{pmatrix} \left( \alpha + \frac{r(n-1)}{2} \right) r^{i-1} \times
\frac{\alpha - 1}{2} \begin{pmatrix} n-1 \\ i-1 \end{pmatrix} r^{i-1} - 1 \nonumber \\
&= \alpha - 1 - (3 \alpha - 1) r(n-1) + r^2 (n-1)^2 - (\alpha - 1) \left( \alpha + \frac{r(n-1)}{2} \right) G,
\label{eq:coneExpandingLProof10}
\end{align}
where
\begin{equation}
G = \sum_{i=2}^{n-1} \left( \begin{pmatrix} n-1 \\ i-1 \end{pmatrix} r^{i-1} \right)^2.
\label{eq:G}
\end{equation}
Next we drop the positive $r^2$-term in \eqref{eq:coneExpandingLProof10}
and use $r(n-1) \le \frac{1}{10}$, by \eqref{eq:rC}, to obtain
\begin{equation}
H > \alpha - 1 - (3 \alpha - 1) r(n-1) - (\alpha - 1) \left( \alpha + \tfrac{1}{20} \right) G.
\nonumber
\end{equation}
In Appendix \ref{app:proofs} (Lemma \ref{le:G}) we show $G < 2 r^2 (n-1)^2$.
So by using $r(n-1) \le \frac{1}{10}$ once more we obtain
\begin{equation}
H > \alpha - 1 - \left( 3 \alpha - 1 + \tfrac{1}{10} (\alpha - 1) \left( 2 \alpha + \tfrac{1}{10} \right) \right) r(n-1).
\nonumber
\end{equation}
Thus $H > 0$ by Lemma \ref{le:HProofEnd}, where
the condition $r(n-1) \le \frac{1}{10 \alpha}$ is a consequence of \eqref{eq:rC}.
\end{proof}

%-------------------------------------------------------------------------------
\subsection{Final arguments}
\label{sub:PsiFinal}

%...................................................................................................
\begin{proof}[Proof of Proposition \ref{pr:cone}]
The result follows immediately from Propositions \ref{pr:coneInvariantL} and \ref{pr:coneExpandingL}.
\end{proof}

%...................................................................................................
\begin{proof}[Proof of Theorem \ref{th:main}]
By Proposition \ref{pr:forwardInvariantRegion} $f$ has a trapping region $\Omega_\nu$.
Thus $f$ has a topological attractor $\Lambda \subset \Omega_\nu$.
By Proposition \ref{pr:cone} there exists a cone $\Psi \subset \mathbb{R}^n$,
that contains more than just the zero vector, and an expansion factor $c > 1$
such that $A_Z v \in \Psi$ and $\| A_Z v \| \ge c \| v \|$ for all $v \in \Psi$ and $Z \in \{ L, R \}$.

Now choose any $x \in \Lambda$ whose forward orbit under $f$ does not intersect $\Sigma$.
Let $v \in \Psi \setminus \{ \bO \}$; to complete the proof we show
\begin{equation}
\limsup_{n \to \infty} \frac{1}{n} \ln \left( \left\| \rD f^n(x) v \right\| \right) > 0.
\label{eq:mainProof1}
\end{equation}
Notice
\begin{equation}
\rD f^n(x) = A_{\chi_{n-1}} \cdots A_{\chi_1} A_{\chi_0} v,
\label{eq:mainProof2}
\end{equation}
where $\chi_i = L$ if $f^i(x)_1 < 0$ and $\chi_i = R$ if $f^i(x)_1 > 0$.
Since the cone is invariant,
each multiplication in \eqref{eq:mainProof2} increases the norm of the vector by a factor of at least $c$.
So $\left\| \rD f^n(x) v \right\| \ge c^n \| v \|$ for all $n \ge 1$,
hence $\frac{1}{n} \ln \left( \left\| \rD f^n(x) v \right\| \right) \ge \ln(c) + \frac{\| v \|}{n}$,
which verifies \eqref{eq:mainProof1}.
\end{proof}

%===============================================================================
\section{Border-collision bifurcations}
\label{sec:bcb}

A border-collision bifurcation occurs when a fixed point collides
with a switching manifold for a piecewise-smooth map.
In applications such bifurcations often arise as grazing bifurcations
of piecewise-smooth ODE systems formulated in the context of a Poincar\'e map \cite{DiBu08}.

Here we consider border-collision bifurcations for which the map is continuous
and the pieces of the map on either side of the switching manifold
can be smoothly extended through the switching manifold.
We show that if the piecewise-linear approximation to the map about the bifurcation
is equivalent (under a coordinate transformation) to the BCNF satisfying the
conditions of Theorem \ref{th:main},
then the border-collision bifurcation creates a robust chaotic attractor, locally.
This is because the terms omitted to form the piecewise-linear approximation
manifest as small nonlinear perturbations to the pieces of the BCNF,
and trapping regions and contracting-invariant expanding cones are robust to such perturbations.
A similar result was obtained in \cite{SiGl24}
for the two-dimensional setting.

Since we concern ourselves only with the local dynamics associated with the bifurcation,
any other pieces of the map can be ignored.
Hence it suffices to consider two-piece maps
\begin{equation}
g(y;\nu) = \begin{cases}
g_L(y;\nu), & h(y) \le 0, \\
g_R(y;\nu), & h(y) \ge 0,
\end{cases}
\label{eq:generalMap}
\end{equation}
where $y \in \mathbb{R}^n$ is the state variable and $\nu \in \mathbb{R}$ is a parameter.
We assume $g_L$, $g_R$, and $h$ are $C^1$,
and $g_L(y;\nu) = g_R(y;\nu)$ at all points for which $h(y) = 0$.

Now suppose $g$ has a border-collision bifurcation at $y = \bO$ when $\nu = 0$.
That is, $g(\bO;0) = \bO$ and $h(\bO) = 0$.
Let
\begin{equation}
\sigma = \nabla h(\bO)
\label{eq:sigma}
\end{equation}
be the gradient vector of $h$ at $y = \bO$.
If $\sigma \ne \bO$ then $h(y) = 0$ locally defines a smooth switching manifold
by the regular value theorem \cite{Hi76}.
Also let
\begin{align}
\hat{A}_L &= \rD g_L(\bO;0), &
\hat{A}_R &= \rD g_R(\bO;0), &
\hat{b} &= \frac{\partial g_L}{\partial \nu}(\bO;0) = \frac{\partial g_R}{\partial \nu}(\bO;0),
\end{align}
denote the derivatives of $g_L$ and $g_R$ at the bifurcation.
For each $Z \in \{ L, R \}$,
if $1$ is not an eigenvalue of $\hat{A}_Z$ then $g_Z$ has a fixed point
\begin{equation}
y^Z = \big( I - \hat{A}_Z \big)^{-1} \hat{b} \nu + \co(\nu),
\label{eq:yZ}
\end{equation}
where $\co(\nu)$ denotes terms that vanish faster than $\nu$.
In what follows we write $B_t$ for the closed ball of radius $t > 0$ centred at $\bO$:
\begin{equation}
B_t = \left\{ y \in \mathbb{R}^n \,\middle|\, \| y \| \le t \right\}.
\label{eq:ball}
\end{equation}

%...................................................................................................
\begin{theorem}
For a continuous piecewise-$C^1$ map \eqref{eq:generalMap} with $g(\bO;0) = 0$ and $h(\bO) = 0$, suppose
\begin{enumerate}
\renewcommand{\labelenumi}{\arabic{enumi})}
\item
\label{it:1}
$\hat{A}_L$ does not have an eigenvector $v \in \mathbb{C}^n$ with $\sigma^{\sf T} v = 0$,
\item
\label{it:2}
$\sigma^{\sf T} \big( I - \hat{A}_R \big)^{-1} \hat{b} > 0$,
\item
\label{it:3}
$\hat{A}_L$ and $\hat{A}_R$ satisfy Assumption \ref{as:main} and \eqref{eq:rA}--\eqref{eq:rC}.
\end{enumerate}
Then there exists $\delta > 0$ and $s > 0$ such that for all $0 < \nu < \delta$
the map $g$ has an attractor with a positive Lyapunov exponent in $B_{\delta s}$.
\label{th:hots}
\end{theorem}

Condition (1) ensures that the piecewise-linear approximation to $g$ is `observable'
and can be converted to the $n$-dimensional BCNF by an affine change of coordinates \cite{Si16}.
Condition (2) ensures $\nu$ unfolds the border-collision bifurcation in a generic fashion
and that $\nu > 0$ corresponds to \eqref{eq:bcnf} with \eqref{eq:bcnfALARb}.
Condition (3) provides the eigenvalue assumptions of Theorem \ref{th:main}.

To illustrate Theorem \ref{th:hots} consider the map
\begin{equation}
y \mapsto
\begin{bmatrix}
0.03 y_3 \\
0.03 y_1 - y_2^2 + \nu \\
-1.2 |y_1 + y_2 + y_3|,
\end{bmatrix}
\label{eq:exampleMap}
\end{equation}
This map can be put in the form \eqref{eq:generalMap}
using the linear switching function $h(y) = \sigma^{\sf T} y$
with $\sigma^{\sf T} = \begin{bmatrix} 1 & 1 & 1 \end{bmatrix}$.
It is readily seen that \eqref{eq:exampleMap} satisfies Conditions (1) and (2) of Theorem \ref{th:hots}.
The matrices
\begin{align}
\hat{A}_L &= \begin{bmatrix} 0 & 0 & 0.03 \\ 0.03 & 0 & 0 \\ 1.2 & 1.2 & 1.2 \end{bmatrix}, &
\hat{A}_R &= \begin{bmatrix} 0 & 0 & 0.03 \\ 0.03 & 0 & 0 \\ -1.2 & -1.2 & -1.2 \end{bmatrix},
\nonumber
\end{align}
have eigenvalues $\alpha \approx 1.230$, $\lambda^L_{2,3} = -0.01499 \pm 0.02556 \ri$,
and $-\beta \approx -1.170$, $\lambda^R_{2,3} = -0.01499 \pm 0.02643 \ri$, respectively,
so Condition (3) is satisfied using $r = 0.031$.
Thus \eqref{eq:exampleMap} has a chaotic attractor for small $\nu > 0$ by Theorem \ref{th:hots}.
This agrees with the numerically computed bifurcation diagram shown in Fig.~\ref{fig:bifurcation}.
For small $\nu > 0$ the chaotic
attractor has two connected components which is consistent
with the attractor of the corresponding skew tent map
that consists of two disjoint intervals.
For larger values of $\nu$ the attractor has a different geometry and topology
due to the $y_2^2$-term in \eqref{eq:exampleMap}.

\begin{figure}[h]
\centering
\includegraphics[width=0.7\linewidth]{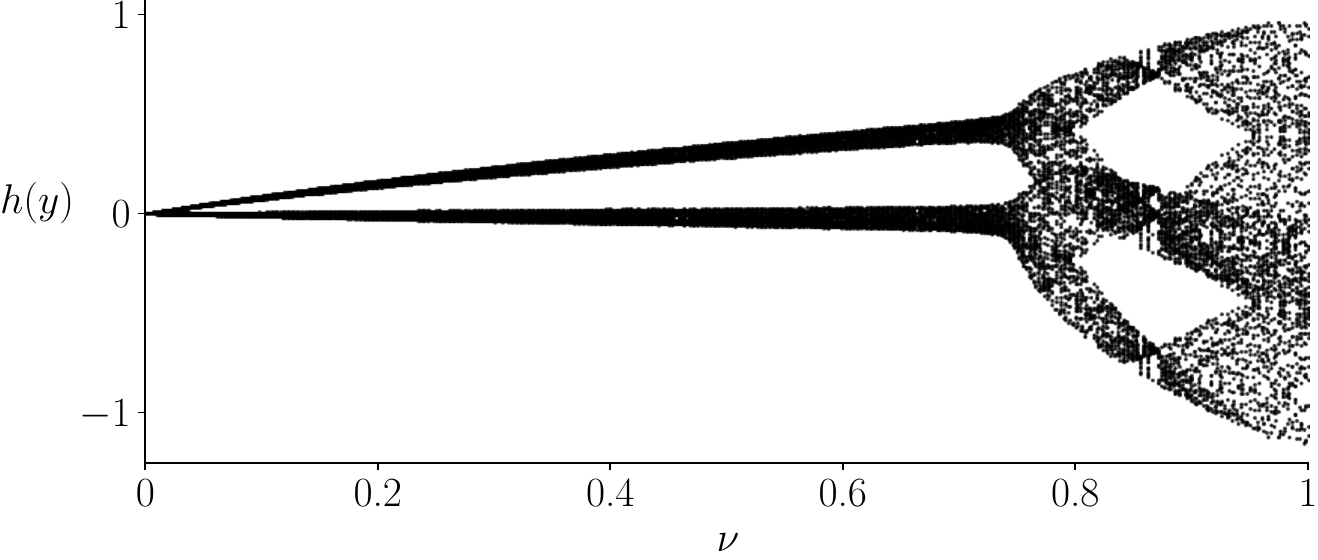}%&
\caption{
A numerically computed bifurcation diagram of
\eqref{eq:exampleMap} showing values of $h(y) = y_1 + y_2 + y_3$
on the vertical axis.}\label{fig:bifurcation}
\end{figure}

%...................................................................................................
\begin{proof}[Proof of Theorem \ref{th:hots}]
Condition (2) implies $\sigma \ne \bO$,
so the switching manifold is smooth and there exists a near-identity
coordinate transformation that linearises this manifold.
By Condition (1) there exists a subsequent affine coordinate transformation
converting the map to
\begin{equation}
x \mapsto \begin{cases}
A_L x + b \mu + E_L(x;\mu), & x_1 \le 0, \\
A_R x + b \mu + E_R(x;\mu), & x_1 \ge 0,
\end{cases}
\label{eq:hotsProof10}
\end{equation}
where $A_L$, $A_R$, and $b$ have the form \eqref{eq:bcnfALARb}
and the error terms $E_L$ and $E_R$ are $\co \left( \| x \| + |\mu| \right)$, see \cite{Di03,Si16}.

We can assume the coordinate transformation is done so that $h(y) \le 0$ [$h(y) \ge 0$]
corresponds to $x_1 \le 0$ [$x_1 \ge 0$].
Then $\mu > 0$ [$\mu < 0$] corresponds to $\nu > 0$ [$\nu < 0$] by Condition (2).

To the map \eqref{eq:hotsProof10} with $\mu > 0$ we perform the spatial blow-up $u = \frac{x}{\mu}$ transforming it to
\begin{equation}
q(u;\mu) = \begin{cases}
A_L u + b + \tfrac{1}{\mu} E_L(\mu u,\mu), & u_1 \le 0, \\
A_R u + b + \tfrac{1}{\mu} E_R(\mu u,\mu), & u_1 \ge 0.
\end{cases}
\label{eq:hotsProof20}
\end{equation}
By dropping the error terms from \eqref{eq:hotsProof20} we are left with the BCNF $f$.
The matrices $A_L$ and $A_R$ have the same eigenvalues as $\hat{A}_L$ and $\hat{A}_R$, respectively,
because these two sets of eigenvalues are the stability multipliers
associated with the fixed point $\bO$ of each piece of the map, and these are independent of coordinates.
Indeed the coordinate transformations induce similarity transformations on the Jacobian matrices
hence leave the eigenvalues unchanged.
By Propositions \ref{pr:forwardInvariantRegion} and \ref{pr:cone}, $f$ has a trapping region $\Omega_{\rm trap}$
and a cone $\Psi$ that is invariant and expanding for $\{ A_L, A_R \}$
and contains a non-zero vector.

Let $t > 0$ be such that $\Omega_{\rm trap} \subset B_t$, and consider again the map \eqref{eq:hotsProof20}.
Since $E_L$ and $E_R$ are $\co \left( \| x \| + |\mu| \right)$,
over $B_t$ the difference between $q$ and $f$ tends to zero as $\mu \to 0$.
Thus there exists $\delta_1 > 0$ such that
for all $0 < \mu < \delta_1$ the set $\Omega_{\rm trap}$ is a trapping region for \eqref{eq:hotsProof20}.
For all such $\mu$ the map $q$ has an attractor $\Lambda_\mu \subset B_t$.

The error terms $E_L$ and $E_R$ are $C^1$ because each piece of $g$ is $C^1$.
Thus at any point $u$ with $u_1 < 0$ [$u_1 > 0$]
the derivative $\rD q(u)$ tends to $A_L$ [$A_R$] as $\mu \to 0$.
Thus there exists $\delta_2 > 0$ such that for all $0 < \mu < \delta_2$,
the cone $\Psi$ is invariant and expanding with expansion factor $c > 0$ for the collection
$\left\{ \rD q(u;\mu) \,\middle|\, u \in B_t,\, u_1 \ne 0 \right\}$.

Choose any $0 < \mu < \min[\delta_1,\delta_2]$ and $x \in \Lambda_\mu$
whose forward orbit under $q$ does not intersect the switching manifold.
Let $v \in \Psi \setminus \{ \bO \}$ and observe
\begin{equation}
\rD q^n(x;\mu) v = \rD q \left( q^{n-1}(x); \mu \right) \cdots \rD q(q(x);\mu) \rD q(x;\mu) v,
\nonumber
\end{equation}
so $\left\| \rD q^n(x;\mu) v \right\| \ge c^n \| v \|$ for all $n \ge 1$.
Thus $\Lambda_\mu \subset B_t$ is an attractor of $q$ and chaotic in the sense of having a positive Lyapunov exponent.
Thus the scaled set $\Xi_\mu = \left\{ \mu u \,|\, u \in \Lambda_\mu \right\} \subset B_{\mu t}$
is a chaotic attractor of \eqref{eq:hotsProof10}.
By mapping $\Xi_\mu$ backwards under the two coordinate transformations used to convert $g$ to \eqref{eq:hotsProof10},
we obtain a chaotic attractor of $g$ in $B_{\mu \omega t}$,
where $\omega > 0$ is a bound on the amount by which vectors expand under the inverse of these transformations.
\end{proof}

%===============================================================================
\section{Discussion}
\label{sec:discussion}

Previous studies of chaos in the BCNF with $n > 2$ have primarily been numerical \cite{DeDu11,Pa18,PaBa18}. In this paper, we have performed rigorous constructions to identify a region of parameter space where the BCNF has a chaotic attractor.
This region is open, thus the chaos is robust for the BCNF.
We have further shown that the chaotic attractor persists under nonlinear perturbations to the pieces of the maps and is also robust in a much broader context.

All calculations needed to obtain the results have been provided in full. To extend the results to other parameter regions, such as those with $\alpha < 1$, computer-assisted proofs will instead likely be preferred \cite{GlSi22b,Si23e}.

Our main result, Theorem \ref{th:main}, has three bounds \eqref{eq:rA}, \eqref{eq:rB}, and \eqref{eq:rC}.
The bound \eqref{eq:rC} contains the coefficient $\frac{1}{10}$ which comes from the inequality \eqref{eq:Qspart1Proof50}. Specifically \eqref{eq:Qspart1Proof50} requires $\ee < \ee^*$, where $\ee^* \approx 0.1073$ is a root of $9 \ee^3 + 84 \ee^2 + 140 \ee - 16$, and we have trimmed $\ee^*$ to $\frac{1}{10}$ for simplicity. The bounds \eqref{eq:rA} and \eqref{eq:rB} contain the coefficient $\frac{3}{7}$ which arises in the proof of Proposition \ref{pr:coneInvariantL}.
This coefficient cannot be more than $\frac{\min[\alpha,\beta]}{2 J}$, where $J$ is the coefficient in \eqref{eq:coneInvariantLProof21}.
Here we are using $\ee = \frac{1}{10}$, so $\frac{\min[\alpha,\beta]}{2 J} = \frac{\alpha \left( \alpha + \frac{4}{5} \right)}{2 \left( 3 \alpha - \frac{9}{10} \right)}$, in the case $\alpha \le \beta$, which equals $\frac{3}{7}$ when $\alpha = 1$.

Theorem \ref{th:main} assumes each piece of the map has one direction of instability. In contrast, Glendinning \cite{Gl15b} assumed all directions are unstable. Cases with $1 < k < n$ directions of instability are more challenging and remain for future work. For smooth families of maps it is known such cases can yield robust chaos, in this context termed `wild chaos' \cite{HiKr13,JeOs22}.

Numerical explorations suggest that the chaotic attractor of the BCNF is often the closure of the unstable manifold of the fixed point $X$. It remains to establish this rigorously, as in \cite{GlSi21}, and identify parameter regions for which the map has no other attractors, thus providing a stronger notion of robustness.
Also it remains to find and characterise bifurcations where the chaotic attractor undergoes changes to its topology when $n \ge 3$.

\section*{Acknowledgements}
This work was supported by Marsden Fund contracts MAU1809 and MAU2209 managed by Royal Society Te Ap\={a}rangi. The authors thank Paul Glendinning and Robert McLachlan for helpful conversations.

\appendix

%===============================================================================
\section{Additional algebraic estimates}
\label{app:proofs}

%...............................................................................
\begin{lemma}
Let $\alpha, \beta > 1$ and $J = \frac{2 \alpha + \min[\alpha,\beta] - \frac{9}{10}}{2 \alpha - \min[\alpha,\beta] + \frac{4}{5}}$.
Then $\frac{3}{7} < \frac{\min[\alpha,\beta]}{2 J}$.
\label{le:boundForContractingInv}
\end{lemma}

%...............................................................................
\begin{proof}
If $\alpha \ge \beta$ then
\begin{equation}
7 (\beta - 1)(\alpha - 1 + \alpha - \beta) + \tfrac{2}{5} (\alpha - \beta) + \tfrac{8}{5} (\alpha - 1) > 0.
\nonumber
\end{equation}
By factoring differently we obtain
\begin{equation}
7 \beta \left( 2 \alpha - \beta + \tfrac{4}{5} \right) - 6 \left( 2 \alpha + \beta - \tfrac{9}{10} \right) > 0,
\nonumber
\end{equation}
which is equivalent to $\frac{3}{7} < \frac{\min[\alpha,\beta]}{2 J}$
in the case $\alpha \ge \beta$.

Also observe
\begin{equation}
(\alpha - 1) \left( 7 (\alpha - 1) + \tfrac{8}{5} \right) > 0.
\nonumber
\end{equation}
By factoring differently we obtain
\begin{equation}
7 \alpha \left( \alpha + \tfrac{4}{5} \right) - 6 \left( 3 \alpha - \tfrac{9}{10} \right) > 0,
\nonumber
\end{equation}
which is equivalent to $\frac{3}{7} < \frac{\min[\alpha,\beta]}{2 J}$ when $\alpha < \beta$.
\end{proof}

%...............................................................................
\begin{lemma}
If $r(n-1) \le 1$ then $G < 2 r^2 (n-1)^2$, where $G$ is given by \eqref{eq:G}.
\label{le:G}
\end{lemma}

%...............................................................................
\begin{proof}
Observe $\begin{pmatrix} n-1 \\ i-1 \end{pmatrix} \le \dfrac{(n-1)^{i-1}}{(i-1)!}$, so
\begin{equation}
G \le \sum_{i=2}^{n-1} \left( \frac{(r(n-1))^{i-1}}{(i-1)!} \right)^2
= \sum_{k=1}^{n-2} \left( \frac{(r(n-1))^k}{k!} \right)^2.
\nonumber
\end{equation}
where we have substituted $k = i-1$.
Removing one of the factorials creates a larger value, as does taking the sum to infinity, thus
\begin{equation}
G < \sum_{k=1}^\infty \frac{(r(n-1))^{2 k}}{k!} = e^{(r(n-1))^2} - 1 < 2 (r(n-1))^2,
\nonumber
\end{equation}
where the last inequality holds because $r(n-1) \le 1$.
\end{proof}

\bibliographystyle{plain}
\bibliography{main}

\begin{thebibliography}{10}

\bibitem{AlMo03}
G.~Alvarez, F.~Montoya, M.~Romera, and G.~Pastor.
\newblock Cryptanalysis of a discrete chaotic cryptosystem using external key.
\newblock {\em Phys. Lett. A}, 319(3-4):334--339, 2003.

\bibitem{BaYo98}
S.~Banerjee, J.A. Yorke, and C.~Grebogi.
\newblock Robust chaos.
\newblock {\em Phys. Rev. Lett.}, 80(14):3049--3052, 1998.

\bibitem{Bu99}
J.~Buzzi.
\newblock Absolutely continuous invariant measures for generic multi-dimensional piecewise affine expanding maps.
\newblock {\em Int. J. Bifurcation Chaos}, 9(9):1743--1750, 1999.

\bibitem{Co02}
W.J. Cowieson.
\newblock Absolutely continuous invariant measures for most piecewise smooth expanding maps.
\newblock {\em Ergod. Th. \& Dynam. Sys.}, 22:1061--1078, 2002.

\bibitem{DeDu11}
S.~De, P.S. Dutta, S.~Banerjee, and A.R. Roy.
\newblock Local and global bifurcations in three-dimensional, continuous, piecewise smooth maps.
\newblock {\em Int. J. Bifurcation Chaos}, 21(06):1617--1636, 2011.

\bibitem{Di03}
M.~di~Bernardo.
\newblock Normal forms of border collision in high dimensional non-smooth maps.
\newblock In {\em Proceedings IEEE ISCAS, Bangkok, Thailand}, volume~3, pages 76--79, 2003.

\bibitem{DiBu08}
M.~di~Bernardo, C.~Budd, A.R. Champneys, and P.~Kowalczyk.
\newblock {\em Piecewise-smooth dynamical systems: theory and applications}.
\newblock Springer-Verlag, New York, 2008.

\bibitem{Ea98}
R.W. Easton.
\newblock {\em Geometric Methods for Discrete Dynamical Systems.}
\newblock Oxford University Press, New York, 1998.

\bibitem{Fr98}
J.~Fridrich.
\newblock Symmetric ciphers based on two-dimensional chaotic maps.
\newblock {\em Int. J. Bifurcation Chaos}, 8(6):1259--1284, 1998.

\bibitem{GhMc23}
I.~Ghosh, R.I. McLachlan, and D.J.W. Simpson.
\newblock Robust chaos in orientation-reversing and non-invertible two-dimensional piecewise-linear maps.
\newblock {\em arXiv preprint arXiv:2307.05144}, 2023.

\bibitem{GhSi22b}
I.~Ghosh and D.J.W. Simpson.
\newblock Robust {D}evaney chaos in the two-dimensional border-collision normal form.
\newblock {\em Chaos}, 32(4):043120, 2022.

\bibitem{Gl15b}
P.~Glendinning.
\newblock Bifurcation from stable fixed point to n-dimensional attractor in the border collision normal form.
\newblock {\em Nonlinearity}, 28(10):3457, 2015.

\bibitem{Gl17}
P.~Glendinning.
\newblock Robust chaos revisited.
\newblock {\em Eur. Phys. J. Special Topics}, 226(9):1721--1738, 2017.

\bibitem{GlSi21}
P.~A Glendinning and D.J.W. Simpson.
\newblock A constructive approach to robust chaos using invariant manifolds and expanding cones.
\newblock {\em Discrete Contin. Dyn. Syst.}, 41(7):3367--3387, 2021.

\bibitem{GlSi22b}
P.A. Glendinning and D.J.W. Simpson.
\newblock Chaos in the border-collision normal form: {A} computer-assisted proof using induced maps and invariant expanding cones.
\newblock {\em Appl. Math. Comput.}, 434:127357, 2022.

\bibitem{GrSw97}
J.~Graczyk and G.~Swiatek.
\newblock Generic hyperbolicity in the logistic family.
\newblock {\em Ann. Math.}, 146(1):1--52, 1997.

\bibitem{GuWi79}
J.~Guckenheimer and R.~F. Williams.
\newblock Structural stability of {L}orenz attractors.
\newblock {\em Publ. Math. IHES}, 50:59--72, 1979.

\bibitem{Hi76}
M.W. Hirsch.
\newblock {\em Differential Topology.}
\newblock Springer-Verlag, New York, 1976.

\bibitem{HiKr13}
S.~Hittmeyer, B.~Krauskopf, and H.M. Osinga.
\newblock Interacting global invariant sets in a planar map model of wild chaos.
\newblock {\em SIAM J. Appl. Dyn. Sys.}, 12(3):1280--1329, 2013.

\bibitem{ItTa79}
S.~Ito, S.~Tanaka, and H.~Nakada.
\newblock On unimodal linear transformations and chaos.
\newblock {\em Proc. Japan Acad.}, 55, Ser. A(71):231--236, 1979.

\bibitem{ItTa79b}
S.~Ito, S.~Tanaka, and H.~Nakada.
\newblock On unimodal linear transformations and chaos {II}.
\newblock {\em Tokyo J. Math.}, 2(2):241--259, 1979.

\bibitem{JeOs22}
H.~Jelleyman and H.M. Osinga.
\newblock Matching geometric and expansion characteristics of wild chaotic attractors.
\newblock {\em Eur. Phys. J. Spec. Top.}, 231(3):403--412, 2022.

\bibitem{Ko01}
L.~Kocarev.
\newblock Chaos-based cryptography: {A} brief overview.
\newblock {\em IEEE Circuits Syst. Mag.}, 1(3):6--21, 2001.

\bibitem{KoLi11}
L.~Kocarev and S.~Lian, editors.
\newblock {\em Chaos-{B}ased {C}ryptography: {T}heory, {A}lgorithms and {A}pplications}.
\newblock Springer, New York, 2011.

\bibitem{KuAl16}
A.~Kumar, S.F. Ali, and A.~Arockiarajan.
\newblock Enhanced energy harvesting from nonlinear oscillators via chaos control.
\newblock {\em IFAC-PapersOnLine}, 49(1), 2016.

\bibitem{LiYo78}
T.~Li and J.A. Yorke.
\newblock Ergodic transformations from an interval into itself.
\newblock {\em Trans. Amer. Math. Soc.}, 235(1):183--192, 1978.

\bibitem{LiDi13}
C.~Liu, A.~Di~Falco, D.~Molinari, Y.~Khan, B.S. Ooi, T.F. Krauss, and A.~Fratalocchi.
\newblock Enhanced energy storage in chaotic optical resonators.
\newblock {\em Nature Photon.}, 7:473--478, 2013.

\bibitem{Lo78}
R.~Lozi.
\newblock Un attracteur {\'e}trange (?) du type attracteur de {H}{\'e}non.
\newblock {\em J. Phys. (Paris)}, 39(C5):9--10, 1978.
\newblock In French.

\bibitem{Ly97}
M.~Lyubich.
\newblock Dynamics of quadratic polynomials, {I-II}.
\newblock {\em Acta. Math.}, 178:185--297, 1997.

\bibitem{MaMa93}
Yu.L. Maistrenko, V.L. Maistrenko, and L.O. Chua.
\newblock Cycles of chaotic intervals in a time-delayed {C}hua's circuit.
\newblock {\em Int. J. Bifurcation Chaos.}, 3(6):1557--1572, 1993.

\bibitem{Me07}
J.D. Meiss.
\newblock {\em Differential Dynamical Systems.}
\newblock SIAM, Philadelphia, 2007.

\bibitem{Mi80}
M.~Misiurewicz.
\newblock Strange attractors for the {L}ozi mappings.
\newblock {\em Nonlinear dynamics, Annals of the New York Academy of Sciences}, 357(1):348--358, 1980.

\bibitem{NuYo92}
H.E. Nusse and J.A. Yorke.
\newblock Border-collision bifurcations including “period two to period three” for piecewise smooth systems.
\newblock {\em Phys. D}, 57(1-2):39--57, 1992.

\bibitem{NuYo95}
H.E. Nusse and J.A. Yorke.
\newblock Border-collision bifurcations for piecewise-smooth one-dimensional maps.
\newblock {\em Int. J. Bifurcation Chaos.}, 5(1):189--207, 1995.

\bibitem{PaPa06}
N.K. Pareek, V.~Patidar, and K.K. Sid.
\newblock Image encryption using chaotic logistic map.
\newblock {\em Image Vis. Comput.}, 24:926--934, 2006.

\bibitem{Pa18}
M.~Patra.
\newblock Multiple attractor bifurcation in three-dimensional piecewise linear maps.
\newblock {\em Int. J. Bifurcation Chaos}, 28(10):1830032, 2018.

\bibitem{PaBa18}
M.~Patra and S.~Banerjee.
\newblock Robust chaos in 3-{D} piecewise linear maps.
\newblock {\em Chaos}, 28(12), 2018.

\bibitem{Ro04}
R.C. Robinson.
\newblock {\em An {I}ntroduction to {D}ynamical {S}ystems. {C}ontinuous and {D}iscrete.}
\newblock Prentice Hall, Upper Saddle River, NJ, 2004.

\bibitem{Ry04}
M.~Rychlik.
\newblock {\em Invariant Measures and the Variational Principle for {L}ozi Mappings.}, pages 190--221.
\newblock Springer, New York, 2004.

\bibitem{Si16}
D.J.W Simpson.
\newblock Border-collision bifurcations in $\mathbb{R}^n$.
\newblock {\em SIAM Rev.}, 58(2):177--226, 2016.

\bibitem{Si23e}
D.J.W. Simpson.
\newblock Detecting invariant expanding cones for generating word sets to identify chaos in piecewise-linear maps.
\newblock {\em J. Difference Eq. Appl.}, 29:1094--1126, 2023.

\bibitem{Si24e}
D.J.W. Simpson.
\newblock Nonsmooth folds are tipping points.
\newblock {\em https://arxiv.org/abs/2406.04587.}, 2024.

\bibitem{SiGl24}
D.J.W. Simpson and P.A. Glendinning.
\newblock Inclusion of higher-order terms in the border-collision normal form: Persistence of chaos and applications to power converters.
\newblock {\em Phys. D}, 462:134131, 2024.

\bibitem{Ts01}
M.~Tsujii.
\newblock Absolutely continuous invariant measures for expanding piecewise linear maps.
\newblock {\em Invent. Math.}, 143:349--373, 2001.

\bibitem{Tu99}
W.~Tucker.
\newblock The {L}orenz attractor exists.
\newblock {\em C. R. Acad. Sci. Paris}, 328(12):1197--1202, 1999.

\bibitem{Tu02}
W.~Tucker.
\newblock A rigorous {ODE} solver and {S}male's 14th problem.
\newblock {\em Found. Comput. Math.}, 2(1):53--117, 2002.

\bibitem{Va10}
S.~van Strien.
\newblock One-parameter families of smooth interval maps: {D}ensity of hyperbolicity and robust chaos.
\newblock {\em Proc. Amer. Math. Soc.}, 138(12):4443--4446, 2010.

\bibitem{Vi14}
M.~Viana.
\newblock {\em Lectures on {L}yapunov {E}xponents.}, volume 145 of {\em Cambridge studies in advanced mathematics}.
\newblock Cambridge University Press, Cambridge, 2014.

\bibitem{WaYo01}
Q.~Wang and L.-S. Young.
\newblock Strange attractors with one direction of instability.
\newblock {\em Commun. Math. Phys.}, 218:1--97, 2001.

\bibitem{WaYo08}
Q.~Wang and L.-S. Young.
\newblock Toward a theory of rank one attractors.
\newblock {\em Ann. Math.}, 167:349--480, 2008.

\bibitem{Yo85}
L.-S. Young.
\newblock Bowen-{R}uelle measures for certain piecewise hyperbolic maps.
\newblock {\em Trans. Amer. Math. Soc.}, 287(1):41--48, 1985.

\end{thebibliography}

\end{document}